\documentclass[8pt,onecolumn]{IEEEtran}
\usepackage{amsmath}
\usepackage{amssymb}
\usepackage{graphics}
\usepackage{epstopdf}
\usepackage{graphicx}
\usepackage{epsfig,amssymb}
\usepackage{amsmath}
\usepackage{mathrsfs}
\usepackage{color}
\usepackage{array}
\usepackage{amsfonts,booktabs}
\usepackage{lipsum,amsmath,multicol}
\usepackage{graphicx}
\usepackage{subfigure}
\usepackage{times}
\usepackage{pdfsync}
\usepackage{pgfplots}
\usepackage{pgfplotstable}
\usepackage[subnum]{cases}
\usepackage{multirow}
\usepackage{algorithm}
\usepackage{algpseudocode}
\usepackage{float}
\usepackage{setspace}
\usepackage{bm}

\usepackage{amsthm}
\usepackage{epstopdf}
\usepackage[center]{caption}
\newtheorem{theorem}{Theorem}
\newtheorem{lemma}{Lemma}

\newcommand{\beq}{\begin{equation}}
\newcommand{\eeq}{\end{equation}}
\newcommand{\beqa}{\begin{eqnarray}}
\newcommand{\eeqa}{\end{eqnarray}}

\newcommand{\paren}[1]{\left(#1\right)}

\begin{document}
	\title{Optimal Privacy-Aware Co-Design of Quantizer and Controller in Networked Control Systems}
	\author{Chuanghong Weng, Ehsan Nekouei \thanks{
			C. Weng and E. Nekouei are with the Department of Electrical Engineering, City University of Hong Kong (e-mail: cweng7-c@my.cityu.edu.hk; enekouei@cityu.edu.hk). 
			
			The work was partially supported by the Research Grants Council of Hong Kong under Project CityU 21208921, a grant from Chow Sang Sang Group Research Fund sponsored by Chow Sang Sang Holdings International Limited.
	}}
	\maketitle
	\begin{abstract}
		This paper investigates the optimal privacy-aware networked control problem, in which the dynamical system affected by a private input process sends its measurement to a remote controller after stochastic quantization. An adversary seeks to infer private system inputs from quantization results and control outputs. The optimal privacy-aware quantizer and controller are obtained by solving a stochastic control problem with mutual information regularization, where the mutual information measures the privacy leakage through the quantizer and controller. We first derive the coupled Bellman equations for the optimal quantizer and controller using the dynamic programming decomposition method. We then analyze the structural properties of the solution, showing that the optimal controller is deterministic, while the optimal quantizer regulates the adversary's belief in a closed-loop manner to enhance privacy. To enable numerical optimization, the quantizer and controller are jointly parameterized and then updated via policy gradient methods, and a binary classification approach is used to approximate privacy leakage. Finally, we validate the effectiveness of the proposed approach through numerical experiments on a building control system.
	\end{abstract}
\subsection{Motivation}
The increasing deployment of networked control systems (NCSs) in smart infrastructures, such as intelligent transportation networks and automated building management systems, has raised significant privacy concerns. These concerns stem from the involvement of third-party computing entities, although legitimately authorized to access sensor measurements and control signals, may exploit this access to infer sensitive information about users or system operations. For example, in automated building control systems, ventilation rates and measured CO$_2$ levels can be analyzed to uncover private occupancy patterns \cite{ebadat2015blind,rueda2020comprehensive}.

Privacy leakage in NCSs can therefore compromise both user confidentiality and system security. To address this, we develop an information-theoretic privacy protection framework for general dynamical systems with private sources, through the co-design of privacy-preserving quantization and control policies. %While privacy-preserving mechanisms for control systems have been extensively studied, most existing literature is confined to linear-Gaussian frameworks, leaving the complexities of non-Gaussian systems relatively unexplored. 
\subsection{Contributions}
We study the privacy-aware design problem for a general networked control system, as illustrated in Fig.~\ref{Fig.NetCtrSys}. In this setting, the system state $X_t$ evolves under the influence of a private input $Y_t$. The measurement of the state $Z_t$ is randomly mapped to an index $S_t$ via a stochastic quantizer, which is then transmitted to a remote controller to generate the control input $U_t$. At the same time, an untrusted third party, referred to as the adversary, has access to both the quantization result and the control input, and may infer the private input process $Y_t$.

To mitigate such privacy leakage, we investigate the joint design of a privacy-aware quantizer and controller that balances control performance and privacy preservation. Privacy leakage is quantified by the mutual information between the private process and the pair consisting of the quantization index $S_t$ and the control signal $U_t$. The resulting co-design problem is formulated as a finite-horizon optimization that minimizes a weighted sum of the privacy leakage and the control cost. The main contributions of this paper are summarized as follows:
\begin{itemize}
	\item \textit{Optimal Privacy-Aware Networked Control via Dynamic Programming Decomposition:} We derive coupled backward optimality equations for the quantizer and controller. The analysis reveals that the optimal quantization policy regulates the adversary's belief about the private input through an inner closed-loop mechanism.
	
	\item \textit{Privacy-Aware Policy Gradient Method:} We jointly parameterize the quantizer and controller and derive a policy-gradient formulation for the privacy-aware networked control objective. In addition, we introduce a classification-based method for approximately computing the privacy leakage and incorporate it into the policy-gradient framework for the joint optimization of the quantization and control policies.
	
	\item \textit{Simulation Study:} We validate the proposed privacy-aware control framework in a building automation scenario, showing that indoor CO$_2$ levels can be effectively regulated while preserving occupancy privacy.
\end{itemize}
\begin{figure}[t]
	\centering
	\includegraphics[width=0.35\textwidth]{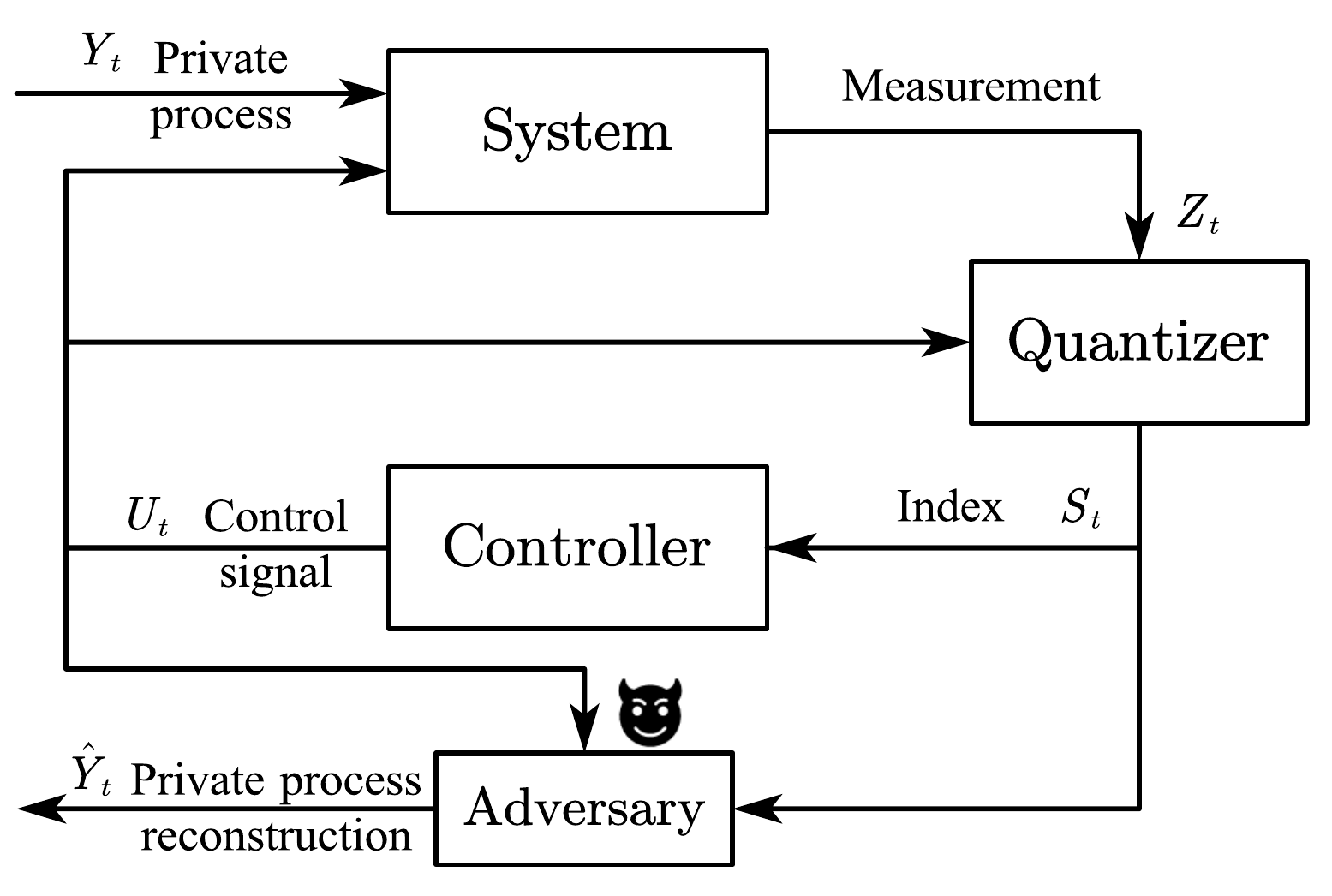}
	\caption{The privacy-aware networked control setup.}\label{Fig.NetCtrSys}
\end{figure}
\subsection{Related work}
Stochastic privacy-preserving approaches for dynamical systems can be broadly classified into differential privacy \cite{dwork2008differential} and information-theoretic privacy \cite{bloch2021overview}. In \cite{le2013differentially}, a differentially private Kalman filter was proposed, together with redesign methods for approximating stable filters under privacy constraints. In \cite{huang2014cost}, a communication strategy was developed to enable agents to exchange information while protecting their internal states. The trade-off between privacy preservation and tracking accuracy in distributed linear control systems, based on the Laplace mechanism, was investigated in \cite{wang2017differential}. More recently, \cite{chen2023differentially} studied privacy in distributed consensus optimization and proposed a differentially private decomposition mechanism that randomly splits the gradient into two sub-states while preserving consensus convergence.

Information-theoretic privacy approaches use mutual information, conditional entropy, and related quantities as privacy metrics, and their applications to dynamical systems have gained attention in recent years \cite{nekouei2019information}. In \cite{tanaka2017directed}, directed information was minimized to reduce privacy leakage in cloud-based linear control systems, and the optimal solution was shown to be a Linear-Gaussian randomized policy obtained via semidefinite programming. Privacy leakage in smart metering systems was analyzed in \cite{li2018information} using mutual information, and a rechargeable battery–based protection mechanism was proposed. In \cite{molloy2023smoother}, the state obfuscation problem in partially observable Markov decision processes (POMDPs) was studied and solved via piecewise-linear optimization techniques \cite{araya2010pomdp}. The work in \cite{zhang2022privacy} addressed privacy leakage of sensitive states through actions in reinforcement learning and developed policy-gradient methods to minimize disclosure while maximizing rewards. Despite these advances, most existing studies on privacy protection in closed-loop control systems are restricted to linear-Gaussian settings. In contrast, this work considers privacy-aware quantizer–controller co-design for non-Gaussian systems. The resulting problem does not fit within the standard POMDP framework \cite{molloy2023smoother}. Instead, it admits a dynamic programming decomposition based on partial information, in which a collection of policies, rather than a single policy, is optimized. Furthermore, in the non-Gaussian setting, mutual information generally does not admit a closed-form expression, which makes conventional convex optimization tools inapplicable. These challenges motivate the development of numerical methods for estimating privacy leakage and solving the quantizer–controller co-design problem.

Beyond differential privacy and information-theoretic privacy, other measures have been proposed, primarily in static settings. These include $R\acute{e}nyi$ min-entropy, maximal leakage, and their variants. The $R\acute{e}nyi$ min-entropy was introduced in \cite{smith2009foundations} to quantify the privacy leakage in deterministic data processing programs. Building on this, \cite{issa2019operational} defined maximal leakage, studied its relation to mutual information, and demonstrated its use in mechanism design. Pointwise maximal leakage, a generalization of maximal leakage, was introduced in \cite{saeidian2023pointwise}, where its connections to other privacy notions were also discussed. In \cite{calmon2013bounds}, a $k$-correlation metric for privacy was proposed, based on the largest $k$ principal inertia components of the joint distribution between private and public random variables. However, applications of these advanced privacy metrics to dynamical systems remain limited.
\subsection{Outline}
The remainder of this paper is organized as follows. Section~\ref{Sec:Prob} presents the system model and formulates the privacy-aware networked control problem. Section~\ref{Sec:Stru} investigates the structural properties of the optimal quantizer and controller. Section~\ref{Sec:Comp} introduces the proposed privacy-aware policy-gradient method for the joint optimization of the quantization and control policies. Section~\ref{Sec:Sim} provides numerical results to demonstrate the effectiveness of the proposed approach. Finally, Section~\ref{Sec:Con} concludes the paper.
\subsection{Notation}
We adopt the following notation throughout the paper. Random variables are denoted by uppercase letters, e.g., $X \in \mathbb{X}$ and $Y \in \mathbb{Y}$, while their realizations are denoted by lowercase letters, e.g., $x$ and $y$. The notation $X^T$ represents the sequence $\left[X_1, X_2, \ldots, X_T\right]$, and $X_t^T$ denotes $\left[X_t, X_{t+1}, \ldots, X_T\right]$. Probability density functions are denoted by $p(\cdot)$, whereas probability mass functions are denoted by $P(\cdot)$. Conditional distributions are written as $p(X \mid Y)$ or $P(X \mid Y)$, depending on whether the variables are continuous or discrete. Moreover, for $t \geq 1$, $p(Y_t \mid X^{t-1})$ denotes the conditional density of $Y_t$ given $X^{t-1}$, with $p(Y_1 \mid X_{0}) \triangleq p(Y_1)$. Similarly, $P(Y_t \mid X^{t-1})$ denotes the corresponding conditional probability mass function, with $P(Y_1 \mid X_0) \triangleq P(Y_1)$. The expectation of a function $f(X)$ is denoted by $\mathsf{E}[f(X)]$. For discrete random variables, we have $\mathsf{E}[f(X)] = \sum_{x \in \mathcal{X}} P(x) f(x)$,
where, for brevity, the domain $\mathcal{X}$ may be omitted, \emph{i.e.}, $\mathsf{E}[f(X)] = \sum_x P(x) f(x)$. For continuous random variables, the expectation is expressed as $\mathsf{E}[f(X)] = \int p(x) f(x) \, dx$. The entropy of a random variable $X$ is defined as $H(X) = -\mathsf{E}[\log P(X)]$. For two random variables $X$ and $Y$, the conditional entropy is $H(X \mid Y) = -\mathsf{E}[\log P(X \mid Y)]$, and the mutual information between $X$ and $Y$ is $I(X; Y) = H(X) - H(X \mid Y)$.
\section{Problem Formulation} \label{Sec:Prob}
We consider a controlled stochastic system with continuous state $X_t \in \mathbb{X} \subseteq \mathbb{R}^{n_x}$, discrete private input $Y_t \in \mathbb{Y} \subseteq \mathbb{R}^{n_y}$, and discrete control input $U_t \in \mathbb{U} \subseteq \mathbb{R}^m$. The system dynamics are governed by the stochastic transition kernel
\begin{align} \label{Eq.Dynamics-state}
	p\left(X_{t+1} \mid X_t, Y_t, U_t\right),
\end{align}
where the private input process $\{Y_t\}$ is modeled as a first-order Markov chain with transition probability
\begin{align}
	P\left(Y_t \mid Y_{t-1}\right).
\end{align}
The system state $X_t$ is not directly observable. Instead, a discrete measurement $Z_t$ is generated according to the observation model
\begin{align} \label{Eq.Dynamics-obs}
	P\left(Z_t \mid X_t\right).
\end{align}
Notably, the system model \eqref{Eq.Dynamics-state}--\eqref{Eq.Dynamics-obs} can be used to study both linear-Gaussian and nonlinear non-Gaussian systems with discrete private input process.

As illustrated in Fig.~\ref{Fig.NetCtrSys}, the quantizer stochastically maps the measurement $Z_t$ to an index $S_t\in \mathbb{S}=\left\{0,1,\cdots, M-1\right\}$, which is then transmitted to the remote controller to determine the control input $U_t$. An adversary, assumed to have complete knowledge of the system dynamics, observes both the quantization index and the control input and seeks to infer the private input $Y_t$. This motivates the need for a privacy-aware co-design of the quantizer and controller. To build intuition, we next present an illustrative example in the context of building control and then develop a general framework for privacy-aware networked control.
\subsection{Motivating Example}
We consider a regularized indoor CO$_2$ concentration control problem \cite{ebadat2015blind,rueda2020comprehensive}. Let the regularization error be the deviation of the indoor CO$_2$ concentration from the target setpoint. The corresponding linearized model of the error state is given by
\begin{align}\label{Eq.Dynamics-CO2}
	X_{t+1} = aX_t + b_yY_t + b_uU_t + W_t,
\end{align}
where $X_t \in \mathbb{R}^n$ denotes the CO$_2$ concentration error state, $Y_t \in \mathbb{R}^{n_y}$ denotes the occupancy input, $U_t \in \mathbb{R}^m$ denotes the ventilation control input, and $W_t \in \mathbb{R}^n$ represents the process noise. The matrices $a \in \mathbb{R}^{n \times n}$, $b_y \in \mathbb{R}^{n \times n_y}$, and $b_u \in \mathbb{R}^{n \times m}$ are the associated system matrices. The observation derived from quantized sensor measurements and the setpoint is approximately modeled as
\begin{align}\label{Eq.Measurement-CO2}
	Z_t = \mathrm{Qua}\left(X_t + V_t\right),
\end{align}
where $Z_t$ denotes the quantized observation of the CO$_2$ concentration error state, $V_t$ is the measurement noise, and $\mathrm{Qua}(\cdot)$ denotes the quantization operator.

To regulate the indoor CO$_2$ concentration around the desired setpoint, the observation $Z_t$ is transmitted to a cloud-based controller, which computes the ventilation command $U_t$. Notably, the cloud service provider has access to both the observation and the control input, and it may act as a curious-but-honest adversary. In particular, it may infer sensitive occupancy information, i.e., the number of occupants and their mobility patterns in the building.

To illustrate this privacy risk, we simulate the indoor CO$_2$ error-state dynamics with $a = 0.85$, $b_y = 0.6$, and $b_u = -0.1$. The process noise $W_t$ and the measurement noise $V_t$ are modeled as independent and identically distributed zero-mean Gaussian random variables with standard deviation $0.1$. Since the measurement noise is relatively small, we simplify the feedback design by assuming that the controller directly uses $Z_t$ to compute the ventilation input, i.e., $U_t = K_p Z_t$ with $K_p = 6.2$. 

Figure~\ref{Fig.EstExp}(a) shows a sample trajectory of the indoor CO$_2$ concentration error state, and Fig.~\ref{Fig.EstExp}(b) shows the corresponding occupancy trajectory. The misdetection events of a maximum-likelihood occupancy estimator are presented in Fig.~\ref{Fig.EstExp}(c), where $\mathbf{1}\{\hat{Y}_t \neq Y_t\}$ equals one whenever the estimator makes an incorrect occupancy estimate. As illustrated in the figure, the estimator can accurately infer the occupancy pattern from the observed error-state trajectory. Such inference capability may be exploited for undesirable purposes, including targeted advertising, surveillance, or even burglary, which motivates the need for a privacy-aware networked control design.
\begin{figure}
	\centering
	\subfigure{
		\centering
		\includegraphics[width=0.30\textwidth]{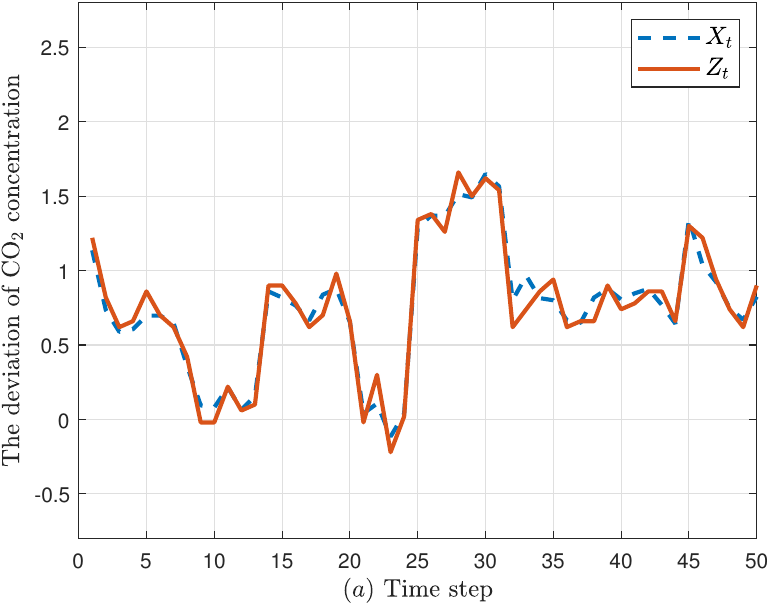}
	}
	\subfigure{
	\centering
	\includegraphics[width=0.30\textwidth]{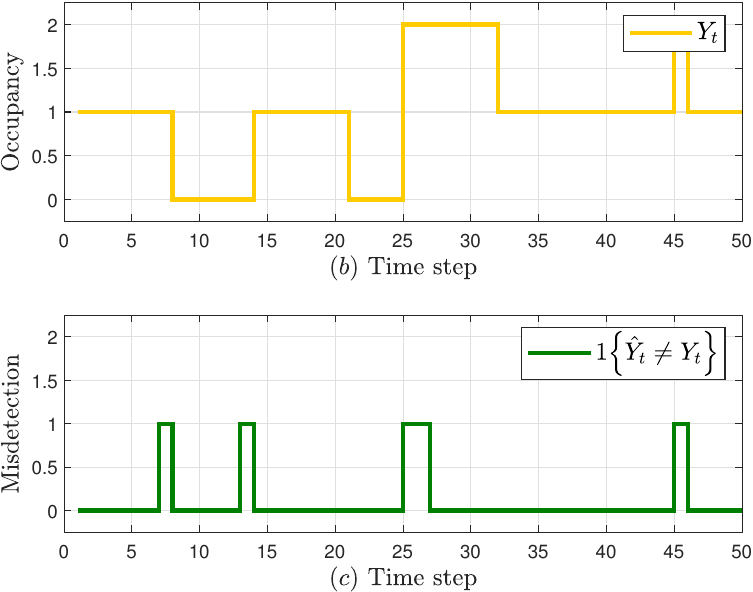}
	}
	\caption{Trajectories of the deviation of CO$_2$ concentration and measurements $(a)$, the occupancy trajectory $(b)$, and the misdetection instances of the occupancy estimator $(c)$.}
	\label{Fig.EstExp}
\end{figure}
\subsection{Optimal Privacy-Aware Quantizer-Controller Co-Design}
We consider the privacy-aware networked control framework illustrated in Fig.~\ref{Fig.NetCtrSys}. The stochastic quantizer generates an index $S_t \in \mathbb{S}$ based on the history of quantized measurements $Z^t=\{Z_1,\ldots,Z_t\}$, past quantization indices $S^{t-1}=\{S_1,\ldots,S_{t-1}\}$, and past control signals $U^{t-1}=\{U_1,\ldots,U_{t-1}\}$. The stochastic quantization policy at time $t$ is denoted by $
\pi_t^e\!\left(S_t \mid Z^t, S^{t-1}, U^{t-1}\right)$, where $\pi_t^e(\cdot \mid Z^t, S^{t-1}, U^{t-1})$ is a conditional probability distribution over the finite set $\mathbb{S}=\{0, 1, \ldots, M-1\}$. For example, $\pi_t^e(S_t=0 \mid Z^t, S^{t-1}, U^{t-1})$ denotes the probability that the quantizer selects the first index at time $t$. When $t=1$, no past decision exists, and the policy reduces to $\pi_1^e(S_1 \mid Z_1)$. The overall quantization policy over the horizon $t=1,\ldots,T$ is then defined as the sequence $\pi^e=\{\pi_t^e\}_{t=1}^T$.  

Similarly, the controller generates the discrete control input $U_t$ according to the conditional probability distribution $\pi_t^c\!\left(U_t \mid S^t, U^{t-1}\right)$, with $\pi_1^c(U_1 \mid S_1)$ for $t=1$. The control policy over the horizon $t=1,\ldots,T$ is defined as the sequence $\pi^c=\{\pi_t^c\}_{t=1}^T$.  

To quantify privacy leakage, we adopt the mutual information between the private input sequence $Y^T$ and the pair $(S^T,U^T)$, defined as follows
\begin{equation}
	I(S^T,U^T;Y^T) 
	= \sum_{s^T,u^T,y^T} P(s^T,u^T,y^T) 
	\log \frac{P(s^T,u^T,y^T)}{P(s^T,u^T)\,P(y^T)},
\end{equation}
which measures the statistical dependence between $Y^T$ and  $(S^T,U^T)$ \cite{cover1999elements}. Intuitively, $I(S^T,U^T;Y^T)$ captures the reduction in the adversary's uncertainty about $Y^T$ due to the disclosure of $S^T$ and $U^T$. According to the definition of mutual information, $I(S^T,U^T;Y^T)=H(Y^T)-H(Y^T|S^T,U^T)\leq H(Y^T)$. In the extreme case where $I(S^T,U^T;Y^T)$ equals the entropy $H(Y^T)$, the private input sequence $Y^T$ can be perfectly inferred by the adversary \cite{cover1999elements}. Therefore, minimizing $I(S^T,U^T;Y^T)$ enhances the privacy of $Y^T$.

In addition, we measure the control performance loss by $c(X_t,U_t)$, with $c(X_t,U_t)=c(X_T)$ when $t=T$. Based on these definitions, the optimal privacy-aware quantizer-controller co-design problem can be formulated as follows
\begin{equation}\label{Eq.OP}
	\min_{\left\{\pi^e, \pi^c\right\}} \mathcal{L}\left(\pi^e, \pi^c\right)= \min_{\left\{ \pi _{t}^e,\pi _{t}^c \right\} _{t=1}^{T}} \sum_{t=1}^{T}{\mathsf{E}\left[ c\left( X_t,U_t \right)\right]}+\lambda I\left( S^T,U^{T};Y^T \right) ,
\end{equation}
where $\lambda$ is a positive constant to penalize the privacy leakage.
\section{Structural Properties of The Optimal Solution}\label{Sec:Stru}
In this section, we analyze the structural properties of the optimal privacy-aware quantizer-controller co-design. We show that the optimal controller is deterministic, while the optimal quantizer admits a hierarchical structure with a local closed-loop mechanism that regulates the adversary's belief about the private information.
\subsection{Deterministic Controller}
We now show that it is sufficient to restrict the controller design to the class of deterministic functions.  

\begin{lemma} \label{Lm.DetCtr}
	The objective function can be reformulated as 
	\begin{align} \label{Eq.simpObj}
		\min_{\pi^e, \pi^c} \mathcal{L}\left(\pi^e, \pi^c\right)=\min_{\left\{ \pi _{t}^e,\pi _{t}^c \right\} _{t=1}^{T}} \sum_{t=1}^T{\mathsf{E}\left[ c\left( X_t,U_t \right) +\lambda c_{i}\left( S_t,Y^{t-1},S^{t-1},U^{t-1} \right) \right]}, 
	\end{align}
	where $c_{i}$ is the information loss
	\begin{align}
		c_{i}\left( S_t,Y^{t-1},S^{t-1},U^{t-1} \right) =\log \frac{P\left( S_t,Y^{t-1} \middle| S^{t-1},U^{t-1} \right)}{P\left( S_t \middle| S^{t-1},U^{t-1} \right) P\left( Y^{t-1} \middle| S^{t-1},U^{t-1} \right)}.
	\end{align}
	Also, the optimal controller can be chosen from the space of deterministic functions, \emph{i.e.},  
	\[
	U_t = \pi_t^c(S^t, U^{t-1}),
	\]
	without loss of optimality.
\end{lemma}
\begin{proof}
	See Appendix~\ref{App:Lm.DetCtr}.
\end{proof}

Although the original optimization problem in \eqref{Eq.OP} allows for stochastic controllers, Lemma~\ref{Lm.DetCtr} establishes that it suffices to consider deterministic policies. This restriction significantly reduces the search space and simplifies the optimization procedure.
\subsection{Equivalent Optimization}
In this subsection, we introduce auxiliary optimization problems for the quantizer and controller and establish their equivalence with \eqref{Eq.OP}. To this end, consider the following collection of conditional distributions,
\begin{align}
	\gamma_t^{e}
	=
	\left\{
	q_{s_t|z^t}:\mathbb{S}\times\mathbb{Z}^t \rightarrow [0,1]
	\,\middle|\,
	\sum_{s_t \in \mathbb{S}} q_t(s_t \mid z^t)=1
	\right\},
\end{align}
where $q_t(\cdot \mid z^t)$ denotes the conditional distribution of $S_t$ given $z^t$. We also define the following collection of measurable functions,
\begin{align}
	\gamma_t^{c}
	=
	\left\{
	\mu_t:\gamma_t^{e}\times\mathbb{S}\rightarrow\mathbb{U}
	\right\},
\end{align}
where $\mu_t$ is a measurable function that depends on the distribution collection $\gamma_t^{e}$ and the quantization index $s_t$, with discrete output $u_t = \mu_t(\gamma_t^{e}, s_t)$. The joint policy collection is then defined as
\begin{align}
	\gamma_t = \left\{\gamma_t^{e},\gamma_t^{c}\right\},
\end{align}
which implicitly depends on $S^{t-1}$ and $U^{t-1}$, i.e., $\gamma_t(S^{t-1},U^{t-1})$. Let the corresponding sequences be denoted by $\gamma^e=\{\gamma_t^e\}_{t=1}^T$, $\gamma^c=\{\gamma_t^c\}_{t=1}^T$, and $\gamma=\{\gamma_t\}_{t=1}^T$. Based on these policy collections, we next formulate two auxiliary optimization problems.

\begin{lemma}\label{Lm.Equivalent}
	The original optimization problem \eqref{Eq.OP} is equivalent to the following auxiliary optimization problems:
	\begin{align}\label{Eq.OP2}
		\min_{\gamma}\mathcal{L}(\gamma)
		=
		\min_{\{\gamma_t\}_{t=1}^{T}}
		\sum_{t=1}^T
		\mathsf{E}
		\!\left[
		c(X_t,U_t)
		+ \lambda c_{i}(S_t,Y^{t-1},S^{t-1},U^{t-1})
		\right],
	\end{align}
	and
	\begin{align}\label{Eq.OP3}
		\min_{\{\gamma^e,\gamma^c\}} \mathcal{L}(\gamma^e,\gamma^c)
		=
		\min_{\{\gamma_t^e,\gamma_t^c\}_{t=1}^{T}}
		\sum_{t=1}^T
		\mathsf{E}
		\!\left[
		c(X_t,U_t)
		+ \lambda c_{i}(S_t,Y^{t-1},S^{t-1},U^{t-1})
		\right],
	\end{align}
	where $\gamma_t^e$ and $\gamma_t^c$ are implicitly allowed to depend on $S^{t-1}$ and $U^{t-1}$.
\end{lemma}

\begin{proof}
	See Appendix~\ref{App:Lm.Equivalent}.
\end{proof}

Lemma~\ref{Lm.Equivalent} shows that the original objective in \eqref{Eq.OP} can be reformulated as a centralized optimization problem in \eqref{Eq.OP2}, based on the common information shared by the quantizer and controller \cite{Nayyar2013,Hubbard2026}, and equivalently as a sequential dynamic game in \eqref{Eq.OP3}. A key feature of both formulations is that the optimization is carried out over the quantization policy collection rather than over an individual quantization policy. This structure arises inherently from the mutual-information regularization. Specifically, the mutual information characterizes the privacy leakage induced by the entire policy collection $\gamma^e$, since it depends on $\gamma^e$ rather than on a single quantization policy $\pi^e$. In particular, the conditional distribution
\begin{align}
	P(S_t,Y^{t-1}\mid S^{t-1},U^{t-1})
	=
	\sum_{z^t}
	\pi_t^e(S_t \mid z^t,S^{t-1},U^{t-1})
	P(z^t,Y^{t-1}\mid S^{t-1},U^{t-1}) \nonumber
\end{align}
appearing in the conditional mutual information is obtained by aggregating all quantization policies associated with the same decision information $(S^{t-1},U^{t-1})$ across different measurement histories $Z^t$, i.e., $q_t\left(\cdot|z^t\right)$ in the collection $\gamma_t^e$.

Therefore, the optimization problem in \eqref{Eq.OP} cannot be solved by determining only the optimal individual quantization distribution $\pi_t^e(\cdot \mid Z^t,S^{t-1}=s^{t-1},U^{t-1}=u^{t-1})$
while the remaining distributions $\pi_t^e(\cdot \mid Z^t,S^{t-1}\neq s^{t-1},U^{t-1}\neq u^{t-1})$
remain unspecified. In other words, to solve \eqref{Eq.OP}, it is sufficient to jointly optimize the quantization policy collection $\gamma^e$ corresponding to the same realizations of $S^{t-1}$ and $U^{t-1}$.
%
%The auxilary decision-making problem over the coding policy collection $\mathcal{A}=\left\{\mathcal{A}_t\right\}_{t=1}^T$ is defined as follows,
%\begin{equation}\label{Eq.OP2}
%	\min_{\left\{\mathcal{A}, \pi^c \right\}} \tilde{\mathcal{L}} \left(\mathcal{A}, \pi^c\right)=\min_{\left\{\mathcal{A}_t, \pi_{t}^c\right\}^T_{t=1}} \sum_{t=1}^{T}{\mathsf{E}\left[ c\left( X_t,U_t \right) \right]} +\lambda I\left( \left. S_t;Y^{t-1} \right| S^{t-1}, U^{t-1}\right),
%\end{equation}
%where we implicitly allow  $\mathcal{A}_{t}$ to  depend on ${S^{t-1}}$ and $U^{t-1}$. The following result establishes the equivalence between the original problem \eqref{Eq.OP} and the auxiliary problem \eqref{Eq.OP2}.  
%
%\begin{lemma} \label{Lm.Equivalent}
%	The original optimization problem \eqref{Eq.OP} is equivalent to the auxiliary optimization problems.
%\end{lemma}
%\begin{proof}
%	See Appendix~\ref{App:Lm.Equivalent}.
%\end{proof}
\subsection{Optimality Equations}
We next derive the coupled Bellman equations of the privacy-aware quantizer and controller based on the optimality principle.
\begin{theorem}\label{Th.OPTEQU}
	The Bellman optimality equation for the joint policy collection optimization \eqref{Eq.OP2} is given by,
	\begin{align}\label{Eq.VFEst2}
		J_{t}^{\star}\left( b_{t}^{e} \right) =\min_{\gamma _t} \int{\sum_{s_t,y^{t-1},z^t}{q _{t}\left( \left. s_t \right|z^t \right) b_{t}^{e}\left( x_t,y^{t-1},z^t \right) \left[ c\left( x_t,\gamma _{t}^{c}\left( \gamma _{t}^{e},s_t \right) \right) +\lambda c_{i}\left( s_t,y^{t-1},\gamma _{t}^{e},b_{t}^{e} \right) +J_{t+1}^{\star}\left( \Phi \left( b_{t}^{e},\gamma _t,s_t \right) \right) \right] dx_t}},
	\end{align}
	with the terminal cost-to-go function $	J_{T+1}^{\star}=0$. The information loss $c_{i}$ is computed via
	\begin{align}
		c_{i}\left( s_t,y^{t-1},\gamma _{t}^{e},b_{t}^{e} \right) =\log \frac{\sum_{z^t}{q _{t}\left( \left. s_t \right|z^t \right) \int{b_{t}^{e}\left( x_t,y^{t-1},z^t \right) dx_t}}}{\left( \sum_{y^{t-1},z^t}{q _{t}\left( \left. s_t \right|z^t \right) \int{b_{t}^{e}\left( x_t,y^{t-1},z^t \right) dx_t}} \right) \left( \sum_{z^t}{\int{b_{t}^{e}\left( x_t,y^{t-1},z^t \right) dx_t}} \right)},
	\end{align}
	and $b^e_t$ is the belief state defined as
	\begin{equation} 
		b_{t}^e\left( x_t,y^{t-1},z^t \right) =p\left( \left. x_t\right|y^{t-1},z^t,U^{t-1} \right) P\left(\left. y^{t-1},z^t\right|S^{t-1},U^{t-1} \right) 
		, 
	\end{equation}
	with $b_{1}^e=p_{\left.x_1\right|z_1}P_{z_1}$ where $p_{\left.X_1\right|Z_1}$ is the observation probability density function of $X_1$ given $Z_1$, and $P_{Z_1}$ is the probability mass function of $Z_1$. The update rule of belief state  is 
	\begin{align}\label{Eq.beliefUpd2}
		b_{t+1}^{e}\left( x_{t+1},y^t,z^{t+1} \right) =\frac{P\left( z_{t+1}|x_{t+1} \right) P\left( y_t|y_{t-1} \right) q _{t}\left( \left. S_t \right|z^t \right) \int{p\left( x_{t+1}|x_t,y_t,\gamma _{t}^{c}\left( \gamma _{t}^{e},S_t \right) \right) b_{t}^{e}\left( x_t,y^{t-1},z^t \right) dx_t}}{\int{\sum_{y^{t-1},z^t}{q _{t}\left( \left. S_t \right|z^t \right) b_{t}^{e}\left( x_t,y^{t-1},z^t \right) dx_t}}}.
	\end{align}
	which is denoted by $b_{t+1}^{e}=\Phi \left( b_{t}^{e},\gamma _t,S_t \right) $.
\end{theorem}
As shown in Theorem~\ref{Th.OPTEQU}, the joint optimization of the policy collections can be formulated as a POMDP with belief state $b_t^e$ and optimal cost-to-go function $J_t^\star$. Based on this formulation, the optimal quantizer and controller can be derived as follows.
\begin{theorem}\label{Th.estCtrPolicies}
	Let $\gamma^{\star}_t\left(b_{t}^e\right)$ denote the solution of \eqref{Eq.VFEst2}. Then, given $\left(Z^t, S^{t-1}, U^{t-1}\right)$, the optimal quantization policy at time $t$ is  $
	\pi_t^{e,\star}\left(\cdot \mid Z^t, b_{t}^e\right) = q_t^{\star}\left(\cdot \mid Z^t\right)$, 
	where $q_t^\star \in \gamma^{e,\star}_t\left(b_{t}^e\right)$, and $b_{t}^e$ is the belief state associated with $\left(S^{t-1}, U^{t-1}\right)$. After sampling $S_t$ from $\pi_t^{e,\star}$, the optimal control action is $U_t^\star = \mu^{\star}_t\left(\gamma_t^{e,\star}(b_t^e), S_t\right)$ with $\mu_t^\star \in \gamma^{c,\star}_t\left(b_{t}^e\right)$, i.e., $U_t^\star = \pi_t^{c,\star}(b_{t}^e,\gamma_t^{e,\star}(b_t^e), S_t)$.
\end{theorem}
\begin{proof}
	The result follows directly from Lemma \ref{Lm.Equivalent} and Theorem \ref{Th.OPTEQU}.
\end{proof}
Similarly, we can derive the optimality equations for the auxiliary optimization problem \eqref{Eq.OP3} as follows.
\begin{theorem} \label{Th.OPTEQU2}
	The Bellman optimality equations for the auxiliary optimization problem \eqref{Eq.OP3} are given by 
	\begin{equation} \label{Eq.VFEst}
		\begin{aligned}
			\tilde{J}_{t}^{e,\star}\left( b_{t}^{e} \right) =\min_{\gamma _{t}^{e}} \sum_{s_t,y^{t-1},z^t}{\int{q _{t}\left( \left. s_t \right|z^t \right) b_{t}^{e}\left( x_t,y^{t-1},z^t \right) \left[\lambda c_{i}\left( s_t,y^{t-1},\gamma _{t}^{e},b_{t}^{e} \right) +\tilde{J}_{t}^{c,\star}\left( \Phi ^c\left( b_{t}^{e},\gamma _{t}^e,s_t \right) \right) \right] dx_t}},
		\end{aligned}
	\end{equation} 
	\begin{equation}\label{Eq.VFCtr}
		\tilde{J}_{t}^{c,\star}\left( b_{t}^c \right) =\min_{U_t} \int{\sum_{y^{t-1},z^t}{b_{t}^c\left( x_t,y^{t-1},z^t \right) c\left( x_t,U_t \right)}dx_t}+\tilde{J}_{t+1}^{e,\star}\left( \Phi ^e\left( b_{t}^c,U_t \right) \right),
	\end{equation}
	where $b_{t}^e$ and $b_{t}^c$ are belief states defined as,
	\begin{equation} \label{Eq.encBelief}
		b_{t}^e\left( x_t,y^{t-1},z^t \right) =p\left( \left. x_t\right|y^{t-1},z^t,U^{t-1} \right) P\left(\left. y^{t-1},z^t\right|S^{t-1},U^{t-1} \right) 
		, 
	\end{equation}
	\begin{equation} \label{Eq.ctrBelief}
		b_{t}^c\left( x_t,y^{t-1},z^t \right) =p\left(\left. x_t\right|y^{t-1},z^t,U^{t-1} \right) P\left(\left. y^{t-1},z^t\right|S^t,U^{t-1} \right) 
		, \nonumber
	\end{equation}
	and $\tilde{J}_{t}^{e,\star}\left( b_{t}^e \right)$ and $\tilde{J}_{t}^{c,\star}\left( b_{t}^c \right)$ are optimal cost-to-go functions associated with $b_{t}^e$ and $b_{t}^c$. The terminal cost-to-go function is $\tilde{J}_{T}^{c,\star}\left( b_{T}^{c} \right) =\sum_{y^{T-1},z^T}{\int{b_{T}^{c}\left( x_T,y^{T-1},z^T \right) c\left( x_T \right) dx_t}}
	$. 
	The belief states are recursively computed using the following forward equations,
	\begin{equation} \label{Eq.BSUPEst}
		b_{t}^c\left( x_t,y^{t-1},z^t \right) =\frac{q_t\left(\left. S_t\right|z^t\right) b_{t}^e\left( x_t,y^{t-1},z^t \right)}{\sum_{y^{t-1},z^t}{q_t\left(\left. S_t\right|z^t\right) \int{b_{t}^e\left( x_t,y^{t-1},z^t \right) dx_t}}},
	\end{equation}
	\begin{equation} \label{Eq.BSUPCtr}
		b_{t+1}^{e}\left( x_{t+1},y^t,z^{t+1} \right) =\int{P\left( \left. z_{t+1} \right|x_{t+1} \right) p\left( \left. x_{t+1} \right|x_t,y_t,U_t \right) P\left( \left. y_t \right|y_{t-1} \right) b_{t}^{c}\left( x_t,y^{t-1},z^t \right) dx_t},
	\end{equation}  
	with $b_{1}^e=p_{\left.x_1\right|z_1}P_{z_1}$ and $b_{1}^c=p_{\left.x_1\right|z_1}P_{\left.z_1\right|S_1}$. The update rules are denoted with $b_{t+1}^e=\Phi ^e\left( b_{t}^c,U_{t} \right)$ and $b_{t}^c=\Phi ^c\left( b_{t}^e,\gamma _{t}^e,S_t \right)$.
\end{theorem}
\begin{proof}
	See Appendix \ref{App:Th.OPTEQU}.
\end{proof}
Given Theorem \ref{Th.OPTEQU2}, the quantization policy collection $\gamma^e_t$ and the controller policy $\pi_t^c$ can be optimized via the coupled backward optimization equations \eqref{Eq.VFEst} and \eqref{Eq.VFCtr}. The following theorem characterizes the optimal quantization policy based on Theorem \ref{Th.OPTEQU2}.

\begin{theorem}\label{Theo: EstPol2}
	Let $\gamma^{e,\star}_t\left(b_{t}^e\right)$ denote the solution of \eqref{Eq.VFEst}. Then, given $\left(Z^t, S^{t-1}, U^{t-1}\right)$, the optimal quantization policy at time $t$ is $
	\pi_t^{e,\star}\left(\cdot \mid Z^t, b_{t}^e\right) = q_t^{\star}\left(\cdot \mid Z^t\right) \in \gamma^{e,\star}_t\left(b_{t}^e\right)$, 
	where $b_{t}^e$ is the belief state associated with $\left(S^{t-1}, U^{t-1}\right)$. 
\end{theorem}

\begin{proof}
	The result follows directly from Lemma \ref{Lm.Equivalent} and Theorem \ref{Th.OPTEQU2}.
\end{proof}

Based on the dynamic programming decomposition in Theorem \ref{Theo: EstPol2}, the optimal quantizer can be obtained by first decomposing the available information $\left(Z^t, S^{t-1}, U^{t-1}\right)$ into two components, $\left(S^{t-1}, U^{t-1}\right)$ and $Z^t$. The collection of quantization policies is optimized using $\left(S^{t-1}, U^{t-1}\right)$, and the final optimal quantization policy is selected based on the current measurement sequence $Z^t$.

\subsection{Structural Properties of the Optimal Privacy-Aware Design}
As established in Theorems~\ref{Th.estCtrPolicies} and \ref{Theo: EstPol2}, the quantizer’s probabilistic knowledge based on $\left(S^{t-1},U^{t-1}\right)$ is characterized by the belief state $b_t^e$, whose evolution depends on the realization of $S_t$. Since the system is non-Gaussian, $b_t^e$ is generally non-Gaussian, and its evolution cannot be determined in advance because of the stochastic nature of $S_t$. In other words, the optimal quantizer stochastically influences the evolution of $b_t^e$ through its dependence on $S_t$, while simultaneously selecting the corresponding policy distribution based on the measurement realization $Z_t$ and the belief state $b_t^e$. Consequently, the optimal privacy-aware quantizer design is a stochastic closed-loop control problem in which the state is given by the belief state $b_t^e$.

We illustrate the structural properties of the privacy-aware quantizer and controller in Fig. \ref{Fig.EstStruCmp}(a) and (b), based on Theorems \ref{Th.estCtrPolicies} and \ref{Theo: EstPol2}, respectively. As shown in Fig. \ref{Fig.EstStruCmp}(a), the optimal quantizer characterized in Theorem \ref{Th.estCtrPolicies} generates the index $S_t$ using the belief state $b_t^e$, updated according to \eqref{Eq.beliefUpd2}, together with the measurement history $Z^t$. The controller then uses $S_t$ and $b_t^e$ to generate the control input $U_t$ for regulating the system state $X_t$. Similarly, in Fig. \ref{Fig.EstStruCmp}(b), the quantizer characterized in Theorem \ref{Theo: EstPol2} updates its belief state $b_t^e$ according to \eqref{Eq.BSUPCtr} and transmits the index $S_t$ to the controller, while the controller updates its own belief according to \eqref{Eq.BSUPEst} and generates the control signal $U_t$. Since $b_{t+1}^e$ evolves according to
$b_{t+1}^{e}=\Phi \left( b_{t}^{e},\gamma _t,S_t \right)$ and
$b_{t+1}^{e}=\Phi ^e\left( \Phi ^c\left( b_{t}^{e},\gamma _{t}^{e},S_t \right) ,U_t \right)$, the quantizer effectively controls the evolution of $b_t^e$ through feedback.

Moreover, given the quantization indices and control inputs $\left(S^t,U^t\right)$, the adversary can form a belief about the private input through the posterior distribution
\begin{equation} \label{Eq.ADVINF}
	b_{t}^{a}\left( y^t \right) =P\left( y^t\mid S^t,U^t \right) =\int{\sum_{Z^{t+1}}{b_{t+1}^{e}\left( x_{t+1},y^t,z^{t+1} \right) dx_{t+1}}},
\end{equation}
which enables it to infer the private information. Since $b_{t+1}^e$ is governed by the feedback quantization policy collection $\gamma_t^e(b_t^e)$, the adversary’s belief is also shaped by the quantizer in a closed-loop manner through the belief update dynamics.

Furthermore, setting $\lambda=0$ in \eqref{Eq.OP} eliminates the privacy-leakage penalty, so that only control performance is taken into account. As a result, the quantizer is no longer required, and the privacy-unaware design shown in Fig.~\ref{Fig.EstStruCmp}(c) is recovered via standard POMDP methods \cite{krishnamurthy2016partially}. In this case, the belief states of both the controller and the adversary evolve passively, without inner-loop regulation by the quantizer. By contrast, the privacy-aware design introduces an additional degree of freedom for regulating the information flow from the measurements to the control actions, thereby allowing the quantizer to actively reduce privacy leakage.

\begin{figure*}[!htpb]
	\centering
	\subfigure[]{
		\centering
		\includegraphics[width=0.42\textwidth]{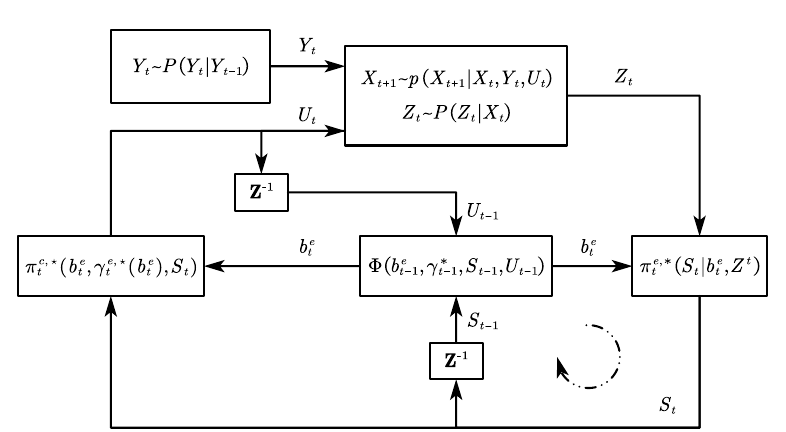}
		\label{Fig.EstStruCmpA}
	}
	\subfigure[]{
		\centering
		\includegraphics[width=0.4\textwidth]{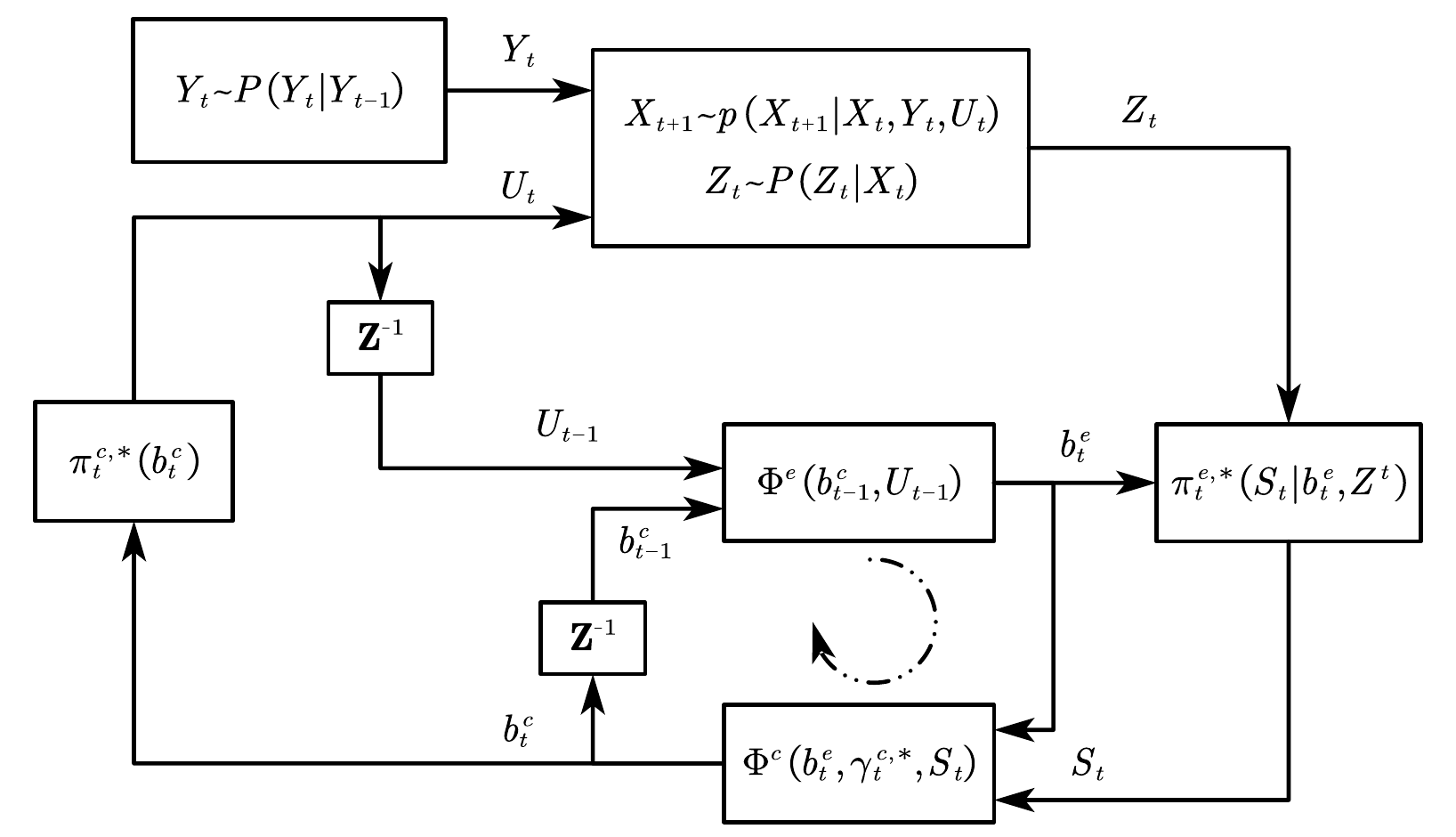}
		\label{Fig.EstStruCmpB}
	}
	\\
	\centering
	\subfigure[]{
		\centering
		\includegraphics[width=0.4\textwidth]{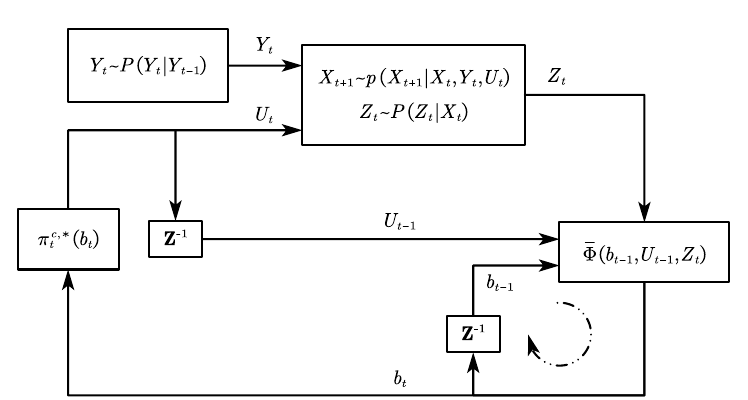}
		\label{Fig.EstStruCmpC}
	}
%	\subfigure[]{
%		\centering
%		\includegraphics[width=0.4\textwidth]{sysStructureGaussian.pdf}
%		\label{Fig.EstStruCmpD}
%	}
	\caption{Structures of networked control systems: (a) privacy-aware design in Theorem \ref{Th.OPTEQU}; (b)  privacy-aware design in Theorem \ref{Th.OPTEQU2}; (c) privacy-unaware design.}% (d) optimal encoder for linear and Gaussian setting.}
	\label{Fig.EstStruCmp}
\end{figure*}
\section{Numerical Optimization}\label{Sec:Comp}
Although the optimal privacy-aware design can, in principle, be obtained through the dynamic programming decomposition techniques developed in Theorems~\ref{Th.OPTEQU} and \ref{Th.OPTEQU2}, exact optimization remains intractable for two main reasons: i) the difficulty of computing the information leakage term $c_{i}\left( S_t,Y^{t-1},S^{t-1},U^{t-1} \right)$, which depends on all possible past measurement realizations $Z^t$, and ii) the curse of dimensionality arising from the continuous state space. To overcome these challenges, we next develop a practical numerical approach for optimizing the privacy-aware quantizer-controller co-design.

\subsection{Policy Gradient Optimization for Joint Design}
Consider the stochastic joint policy $\pi_t\left( \left. U_t,S_t \right|Z^t, S^{t-1},U^{t-1} \right)
=
\pi_t^c\left( \left. U_t \right|S^t,U^{t-1} \right)
\pi_t^e\left( \left. S_t \right|Z^t, S^{t-1},U^{t-1} \right)$. 
We parameterize this joint policy as $
\pi_{\theta}\left( \left. U_t,S_t \right|Z^t, S^{t-1},U^{t-1} \right)$, which can be implemented using a recurrent neural network of the form $
\pi_{\theta}\left( \left. U_t,S_t \right|h_{\theta,t} \right)$, where the recurrent hidden state evolves according to $
h_{\theta,t}=f_{\theta}\left( h_{\theta,t-1}, Z_t, S_{t-1},U_{t-1} \right)$. Under this parameterization, the objective function in \eqref{Eq.OP} can be written as
\begin{align}
	L_{\theta}\left( \pi_{\theta} \right)
	=
	\mathsf{E}_{\theta}\left[
	\sum_{t=1}^T
	\Big(
	c\left( X_t,U_t \right)
	+\lambda c_{i,\theta}\left( S_t,Y^{t-1},S^{t-1},U^{t-1} \right)
	\Big)
	\right].
\end{align}
Since the information leakage term $c_{i,\theta}$ depends on the quantization policy $\pi_\theta^e$, as discussed in Sec.~\ref{Sec:Stru}, it also depends on the parameter $\theta$. This feature differentiates the proposed formulation from standard reinforcement learning problems \cite{Otterlo2012}. We next derive the gradient of $L_{\theta}\left( \pi_{\theta} \right)$ with respect to $\theta$ for numerical optimization.
\begin{theorem}\label{Th.gradientObj}
	The gradient of $L_{\theta}\left( \pi_{\theta} \right)$ with respect to $\theta$ is
	\begin{align}\label{Eq.gradObj}
		\nabla_{\theta}L_{\theta}\left( \pi_{\theta} \right)
		=
		\mathsf{E}_{\theta}\left[
		\left(
		\sum_{t=1}^T
		\Big(
		c\left( X_t,U_t \right)
		+\lambda c_{i,\theta}\left( S_t,Y^{t-1},S^{t-1},U^{t-1} \right)
		\Big)
		\right)
		\left(
		\sum_{t=1}^T
		\nabla_{\theta}\log
		\pi_{\theta}\left( \left. U_t,S_t \right|Z^t, S^{t-1},U^{t-1} \right)
		\right)
		\right].
	\end{align}
\end{theorem}

\begin{proof}
	See Appendix \ref{App:Th.gradientObj}.
\end{proof}

Based on Theorem \ref{Th.gradientObj}, the gradient $\nabla_{\theta}L_{\theta}\left( \pi_{\theta} \right)$ can be approximated using Monte Carlo methods based on collected data trajectories $\tau=\left( Y^T,X^T,Z^T,S^T,U^{T} \right)$, provided that the policy evaluation terms $c$ and $c_i$ can be computed.
\subsection{Information Loss Computation via Binary Classification} \label{Sec.InfApp}
To update the controller and the quantizer, we need to compute the information loss
\begin{equation} \label{Eq.likeliRatio}
	c_{i}\big(S_t,Y^{t-1},S^{t-1},U^{t-1}\big)
	=
	\log \frac{P\big( S_t \mid Y^{t-1},S^{t-1},U^{t-1} \big)}{P\big( S_t \mid S^{t-1},U^{t-1} \big)},
\end{equation}
which depends on the entire quantization policy sequence and becomes computationally intractable for long time horizons or high-dimensional measurement spaces. To avoid the exact evaluation of \eqref{Eq.likeliRatio}, we employ a binary classification approach motivated by \cite{molavipour2021neural}.

Let $\tilde{S}_t$ be a discrete uniform random variable sharing the same support as $S_t$, independent of $S^T$, $U^{T}$, and $Y^T$. We define a binary random variable $M_t \in \{0,1\}$ with prior probabilities $P(M_t=0)=P(M_t=1)=0.5$, and introduce the auxiliary random variable  
\begin{equation}
	\bar{S}_t = M_t S_t + (1-M_t) \tilde{S}_t,
\end{equation}  
which distinguishes between $S_t$ and $\tilde{S}_t$ depending on $M_t$. The following lemma shows that the likelihood ratio \eqref{Eq.likeliRatio} can be obtained via binary classification.

\begin{lemma} \label{Lm.InfCmp}
	The likelihood ratio can be expressed as
	\begin{align} \label{Eq.likeRatio}
		c_{i}\big(S_t,Y^{t-1},S^{t-1},U^{t-1}\big) 
		= \log \frac{w_t^\star(S_t,Y^{t-1},S^{t-1},U^{t-1})}{1 - w_t^\star(S_t,Y^{t-1},S^{t-1},U^{t-1})} 
		- \log \frac{\xi_t^\star(S_t,S^{t-1},U^{t-1})}{1 - \xi_t^\star(S_t,S^{t-1},U^{t-1})},
	\end{align}
	where $w_t^\star$ and $\xi_t^\star$ are the optimal solutions of the following binary classification problems
	\begin{align} \label{Eq.likeRatioOP1} 
		\max_{\{w(\cdot )\}_{t=1}^{T}} \mathsf{E}\Biggl[ \sum_{t=1}^T{\bigl( M_t\log w_t(\bar{S}_t,Y^{t-1},S^{t-1},U^{t-1})+\left( 1-M_t \right) \log \bigl( 1-w_t(\bar{S}_t,Y^{t-1},S^{t-1},U^{t-1}) \bigr) \bigr)} \Biggr] ,
	\end{align}
	and
	\begin{align} \label{Eq.likeRatioOP2}
		\max_{\{\xi (\cdot )\}_{t=1}^{T}} \mathsf{E}\Biggl[ \sum_{t=1}^T{\bigl( M_t\log \xi _t(\bar{S}_t,S^{t-1},U^{t-1})+\left( 1-M_t \right) \log \bigl( 1-\xi _t(\bar{S}_t,S^{t-1},U^{t-1}) \bigr) \bigr)} \Biggr] ,
	\end{align}
	with $w_t$ and $\xi_t$ being measurable functions mapping into $[0,1]$. In particular, $w_t$ is a binary classifier that estimates the probability of $M_t=1$ given $(\bar{S}_t,Y^{t-1},S^{t-1},U^{t-1})$, while $\xi_t$ plays the same role given $(\bar{S}_t,S^{t-1},U^{t-1})$.
\end{lemma}

\begin{proof}
	See Appendix~\ref{App:Th.InfCmp}.
\end{proof}
With Lemma~\ref{Lm.InfCmp}, the likelihood ratio in \eqref{Eq.likeRatio} can be numerically approximated by first parameterizing the classifiers $w$ and $\xi$, then collecting data tuples of $(\bar{S}^T,Y^T,S^T,U^T)$, and finally solving the optimization problems \eqref{Eq.likeRatioOP1} and \eqref{Eq.likeRatioOP2} approximately via stochastic gradient ascent. This procedure provides a tractable approximation of the privacy leakage associated with the index $S_t=s_t$ conditioned on the past decisions $(S^{t-1}=s^{t-1},U^{t-1}=u^{t-1})$.

Finally, based on Theorem \ref{Th.gradientObj} and Lemma \ref{Lm.InfCmp}, we can alternately update the joint policy $\pi$ and the information loss estimator with stochastic gradient descent methods.
\begin{figure}[b]
	\centering
	\includegraphics[width=0.28\textwidth]{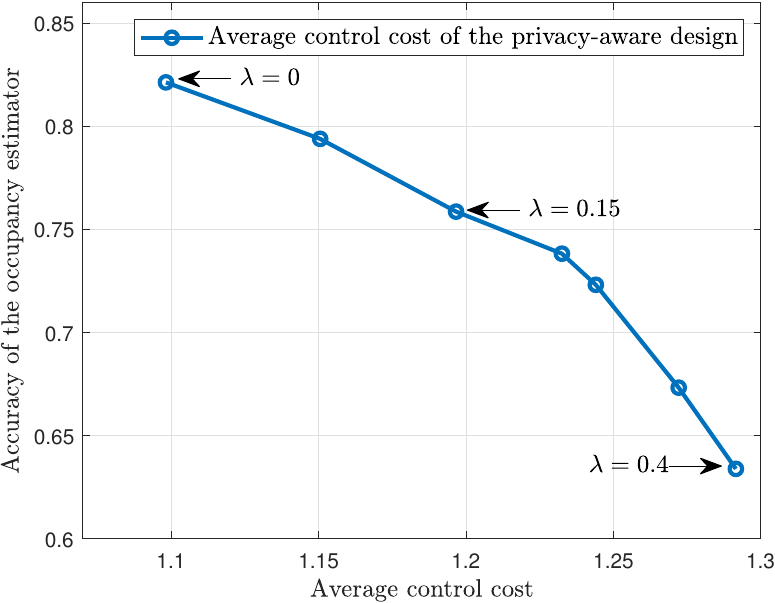}
	\caption{The average stage cost under the optimal privacy-aware quantizer-controller co-design.}
	\label{Fig.Tradeoff}
\end{figure}
\section{Simulation}\label{Sec:Sim}
In this section, we provide numerical validation of the proposed privacy-aware networked control framework using the indoor CO$_2$ concentration regulation model given in \eqref{Eq.Dynamics-CO2}. The model parameters are set to $a = 0.85$, $b_y = 0.6$, and $b_u = -0.1$. The process noise $W_t$ and the observation noise $V_t$ are assumed to be independent and identically distributed zero-mean Gaussian random variables with identical standard deviations of $0.1$.

The private occupancy process $Y_t$ is modeled as a Markov chain with state space $\{0,1,2\}$ and transition probability matrix
\begin{align}
	\begin{bmatrix}
		0.80 & 0.15 & 0.05 \\
		0.10 & 0.85 & 0.05 \\
		0.05 & 0.15 & 0.80
	\end{bmatrix}. \nonumber
\end{align}
The control performance is evaluated using the stage cost $c(x_t,u_t) = x_t^2 + 0.01\,u_t^2$, which penalizes deviations of the indoor CO$_2$ concentration while incorporating a quadratic penalty on the control effort to account for energy consumption. The quantization alphabet is restricted to $S_t \in \{0,1,2,3\}$, and the simulation horizon is set to $T = 50$. Finally, the quantizer is initially optimized via the discrete variational auto-quantizer method with $\lambda=0$ \cite{rolfe2017discrete}, and then jointly optimized with the controller using the proposed privacy-aware numerical optimization method described in Sec.~\ref{Sec:Comp}.

We first investigate the privacy--utility trade-off of the proposed privacy-aware quantizer-controller co-design. Figure~\ref{Fig.Tradeoff} shows the average stage cost $c_t$ versus the accuracy of the maximum-likelihood occupancy estimator. Lower estimation accuracy indicates a higher level of privacy, since the adversary is less able to correctly infer the sensitive occupancy state. As observed in Fig.~\ref{Fig.Tradeoff}, enhanced privacy is obtained at the expense of control performance. In particular, as the privacy penalty $\lambda$ increases from 0 to 0.4, the average control cost rises from 1.098 to 1.291, an increase of approximately 17.6\%, while the adversary’s estimation accuracy drops from 0.821 to 0.634, a decrease of approximately 22.8\%.

To further illustrate this trade-off, Figs.~\ref{Fig.Est1}--\ref{Fig.Est3} show the state trajectories, the corresponding state reconstructions obtained by a maximum likelihood state estimator, the quantization indices, and the misdetection results under different values of the privacy penalty $\lambda$. As shown in Fig.~\ref{Fig.Est2}, the quantization indices become increasingly flattened around $S_t = 2$ as $\lambda$ increases, thereby removing more informative features from the transmitted data. Consequently, the state reconstruction error shown in Fig.~\ref{Fig.Est1} increases, leading to poorer regulation performance. At the same time, Fig.~\ref{Fig.Est3} shows that the accuracy of the occupancy estimator decreases, since the flatter quantization indices also suppress sensitive information related to the indoor occupancy. 
\begin{figure}[h]
	\centering
	\subfigure{
		\centering
		\includegraphics[width=0.28\textwidth]{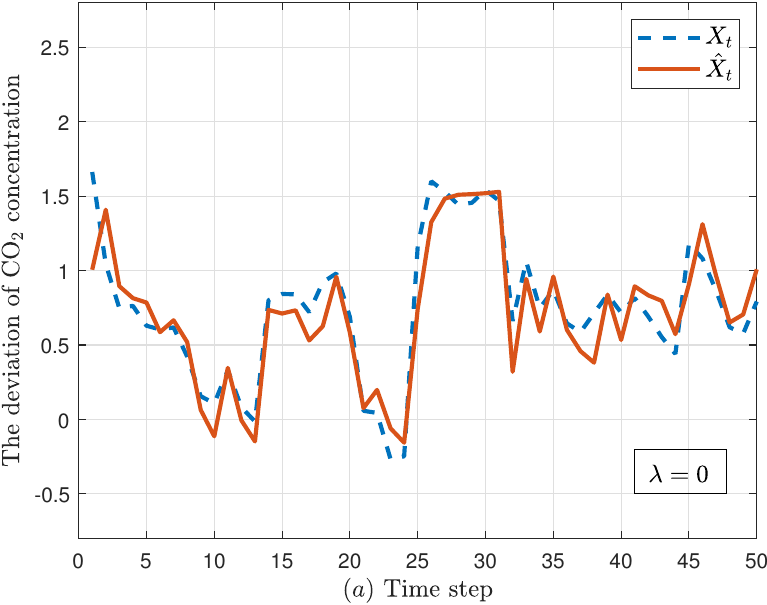}
	}
	\subfigure{
		\centering
		\includegraphics[width=0.28\textwidth]{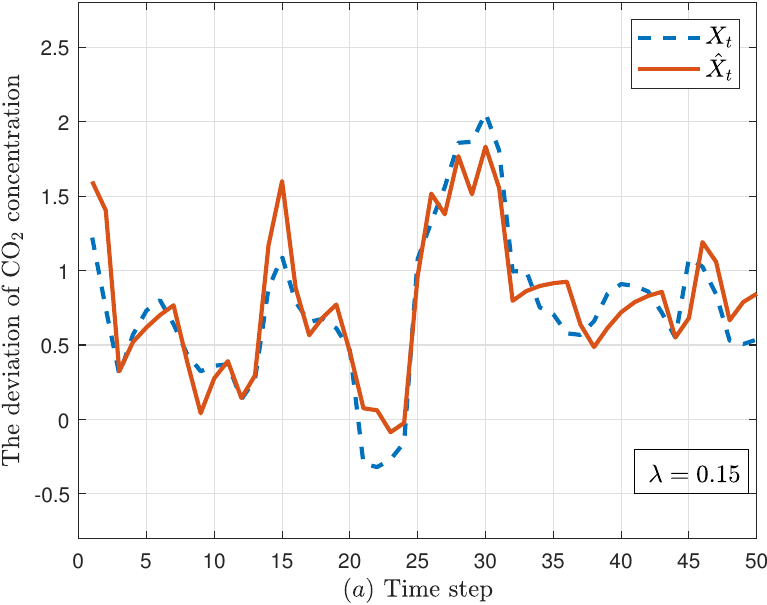}
	}
	\subfigure{
		\centering
		\includegraphics[width=0.28\textwidth]{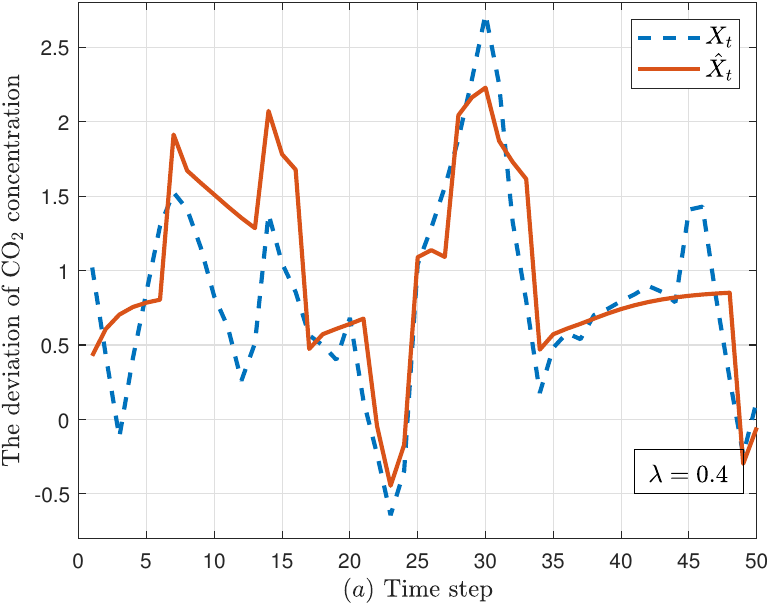}
	}
	\caption{Trajectories of indoor CO$_2$ concentration error and its reconstruction from quantization indices.}
	\label{Fig.Est1}
\end{figure}
\begin{figure}[h]
	\centering
	\subfigure{
		\centering
		\includegraphics[width=0.28\textwidth]{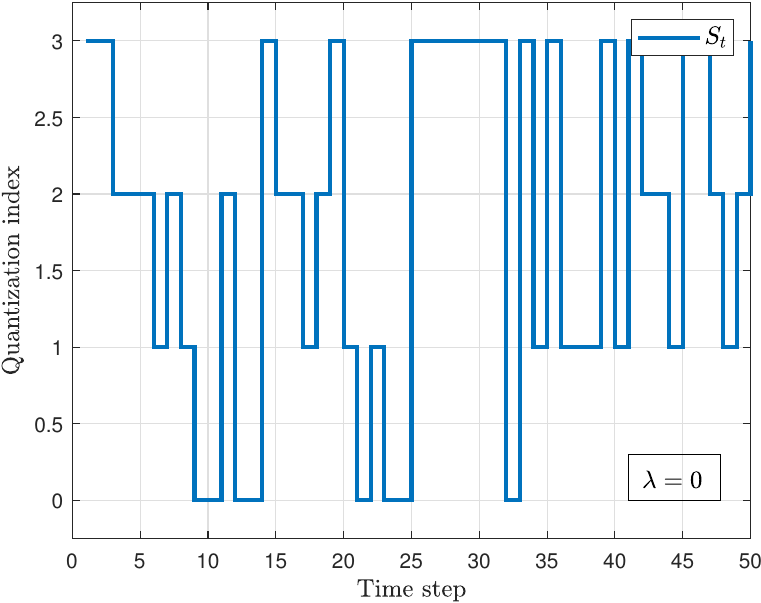}
	}
	\subfigure{
		\centering
		\includegraphics[width=0.28\textwidth]{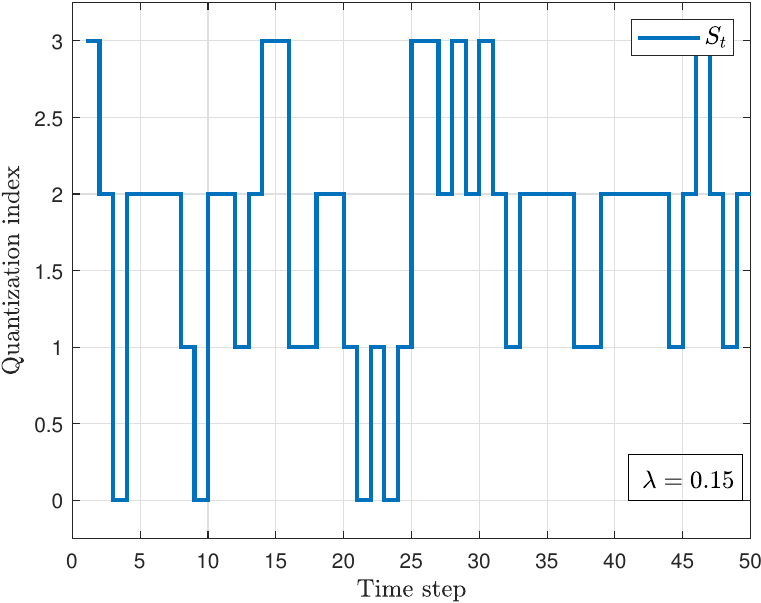}
	}
	\subfigure{
		\centering
		\includegraphics[width=0.28\textwidth]{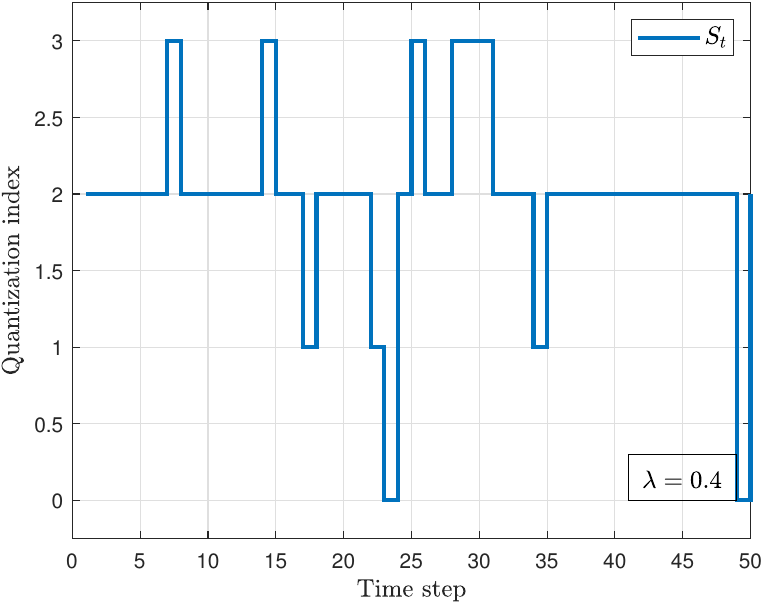}
	}
	\caption{Trajectories of quantization indices.}
	\label{Fig.Est2}
\end{figure}
\begin{figure}[h]
	\centering
	\includegraphics[width=0.28\textwidth]{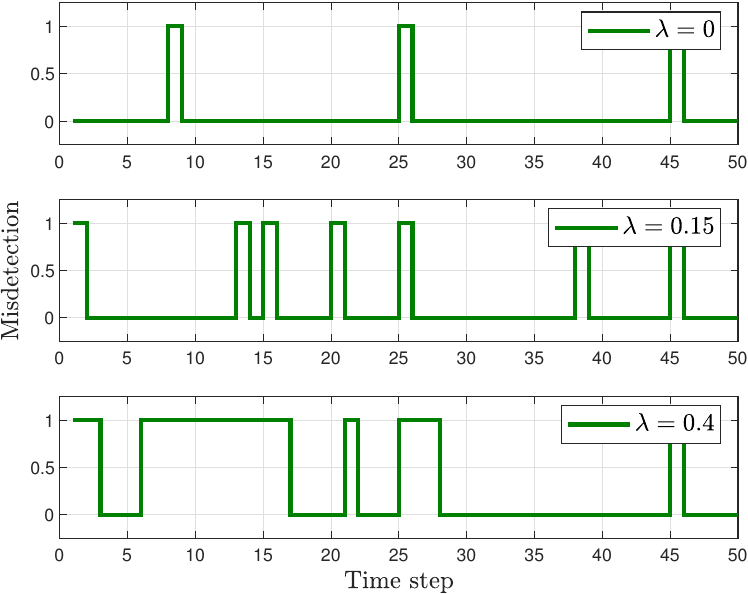}
	\caption{Misdetections of the occupancy estimator.}
	\label{Fig.Est3}
\end{figure}
\section{Conclusion}\label{Sec:Con}
In this paper, we developed an information-theoretic framework for privacy-aware networked control, in which the quantizer and controller are jointly optimized. We characterized the structural properties of the optimal privacy-aware quantizer and controller through the corresponding optimality equations. In addition, we proposed a numerical approach for their joint optimization and validated the proposed framework through an indoor CO$_2$ concentration control example.
%\clearpage

\bibliographystyle{ieeetr}
\bibliography{reference}

\appendices

\section{Proof of Lemma \ref{Lm.DetCtr}}\label{App:Lm.DetCtr}
To prove Lemma \ref{Lm.DetCtr}, we first show that the mutual information can be expanded into an additive form.
\begin{lemma} \label{Lm.MI}
	The mutual information term in \eqref{Eq.OP} can be expanded as 
	\begin{align} \label{Eq.ExpandObj}
		I\left(S^T, U^{T};Y^{T}\right) = \sum_{t=1}^{T} {I\left( \left. S_t;Y^{t-1} \right| S^{t-1}, U^{t-1}\right)}
		=		\sum_{t=1}^{T} {\mathsf{E}\left[c_{i}\left( S_t,Y^{t-1},S^{t-1},U^{t-1} \right)\right]},
	\end{align}
	where $c_{i}\left( S_t,Y^{t-1},S^{t-1},U^{t-1} \right) =\log \frac{P\left( S_t,Y^{t-1} \middle| S^{t-1},U^{t-1} \right)}{P\left( S_t \middle| S^{t-1},U^{t-1} \right) P\left( Y^{t-1} \middle| S^{t-1},U^{t-1} \right)}$ and $S_{0}=U_{0}=\emptyset$.
\end{lemma}
\begin{proof}
	See Appendix \ref{App:Lm.MI}.
\end{proof}
As a result, we have \eqref{Eq.simpObj}. Besides, given $\left(S^t, U^{t-1}\right)$, the stage cost $\mathsf{E}\left[ c\left( X_t,U_t \right) +\lambda c_{i}\left( S_t,Y^{t-1},S^{t-1},U^{t-1} \right)\middle| S^t, U^{t-1} \right]$ is a linear combination of the stochastic control policy. According to \cite{krishnamurthy2016partially}, it is sufficient to consider only the deterministic controller during optimization.
\section{Proof of Lemma \ref{Lm.Equivalent}}\label{App:Lm.Equivalent}
According to Lemma \ref{Lm.MI}, the objective function \eqref{Eq.OP} can be simplified as \eqref{Eq.simpObj}. 

To prove the equivalence between \eqref{Eq.simpObj} and \eqref{Eq.OP2}, we first show that given a quantization policy $\pi^e=\left\{\pi_{t}^e\right\}^T_{t=1}$ and a control policy $\pi_c=\left\{\pi_{t}^c\right\}^T_{t=1}$, we can construct the policy collections $\gamma^e=\left\{\gamma^e_t\right\}^T_{t=1}$ and $\gamma^c=\left\{\gamma^c_t\right\}^T_{t=1}$ for the auxiliary decision-making problem \eqref{Eq.OP2} that achieves the same value of objective function \eqref{Eq.simpObj} as $(\pi^e,\pi^c)$. To this end, given $\pi^e$, we choose the quantization policy and control policy collections as
$$
\gamma^e_t=\left\{ q_t\left( \left. s_t\right|z^t \right)=\pi_{t}^e\left({s_t}\left|{z^t,S^{t-1}}, U^{t-1}\right. \right), \forall s_t, z^t \right\}, \nonumber
$$
and
$$
\gamma^c_t=\left\{ \mu_t\left( \gamma^e_t, s_t \right)=\pi_{t}^c\left({s_t}, {S^{t-1}}, U^{t-1} \right), \forall s_t \right\}, \nonumber
$$
where we collect the conditional distributions associated with same $(S^{t-1}, U^{t-1})$ into the collection $\gamma^e_t$, and then select the control policies with same $(S^{t-1}, U^{t-1})$ into $\gamma^c_t$. Clearly, we have $\mathcal{{L}}\left(\gamma^e, \gamma^c\right)=\mathcal{L}\left(\pi^e, \pi^c\right)$.% where  $\mathcal{{L}}\left(\gamma^e, \pi_u\right)$ and $\mathcal{L}\left(\pi_e, \pi_u\right)$ are the values of objective function under $\paren{\gamma^e, \pi_u}$ and $\paren{\pi_e, \pi_u}$, respectively. 

We next show that, given $\gamma^e$ and $\gamma^c$, we can construct the quantization policy $\pi^e$ and the control policy $\pi^c$ that achieve the same value of objective function as $\gamma^e$ and $\gamma^c$. Given $\gamma^e$ and $\gamma^c$, let $	\pi_{1}^e\left(\left. {S_1}\right|{Z_1}\right)=q_1\left(\left. S_1\right|Z_1\right)$. For $t>1$, given the policy collections $\gamma^e_t\paren{S^{t-1}, U^{t-1}}=\mathcal{P}_t\paren{S^{t-1}, U^{t-1}}$ and $\gamma^c_t\paren{S^{t-1}, U^{t-1}}=\mathcal{Q}_t\paren{S^{t-1}, U^{t-1}}$ via \eqref{Eq.OP2}, we select the quantization policy by,
$$
\pi_{t}^e\left(\left. {S_t}\right|{Z^t,S^{t-1}, U^{t-1}}\right)=q_t\left(\left. S_t\right|Z^t\right)\in \mathcal{P}_t\paren{S^{t-1}, U^{t-1}}.
$$
Then, we randomly sample $S_t$ from $q_t(\cdot | Z^t)$, and choose the control policy as follows,
$$
\pi_{t}^c\left(S^{t}, U^{t-1}\right)=\mu_t\left(\gamma_t^e,  S_t\right)\in \mathcal{Q}_t\paren{S^{t-1}, U^{t-1}}.
$$
Clearly, we have $\mathcal{L}\left(\pi^e, \pi^c\right)=\mathcal{{L}}\left(\gamma^e, \gamma^c\right)$. Thus, the auxiliary decision-making problem \eqref{Eq.OP2} and the simplified optimization problem \eqref{Eq.simpObj} are equivalent.

The equivalence between \eqref{Eq.simpObj} and \eqref{Eq.OP3} can be proved similarly.
\section{Proof of Theorem \ref{Th.OPTEQU}}\label{App:Th.OPTEQU}
The next lemma expresses the stage cost in \eqref{Eq.OP2} in terms of the quantizer's belief state $b_{t}^e$ and presents the recursive update rule of $b_{t}^e$.
\begin{lemma} \label{Lm.BeliefStateLoss}
	Let $b_{t}^e$ denote the quantizer's belief state, \emph{i.e.,} $b_{t}^e\left( x_t,y^{t-1},z^t \right) =p\left( \left. x_t\right|y^{t-1},z^t,U^{t-1} \right) P\left(\left. y^{t-1},z^t\right|S^{t-1},U^{t-1} \right)$.
	Given the policy collection $\gamma^e_t$ and information $\paren{S^{t-1}, U^{t-1}}$, we have 
	\begin{align}\label{Eq.condMIProof}
		\mathsf{E}\left[ \left. c_i\left( S_t,Y^{t-1},S^{t-1},U^{t-1} \right) \right|S^{t-1},U^{t-1} \right] =\sum_{s_t,y^{t-1},z^t}{q_t\left( \left. s_t \right|z^t \right) \int{b_{t}^{e}\left( x_t,y^{t-1},z^t \right) dx_tc_i\left( s_t,y^{t-1},\gamma _{t}^{e},b_{t}^{e} \right)}},
	\end{align}
	with
	\begin{align}\label{Eq.infLossProof}
		c_i\left( s_t,y^{t-1},\gamma _{t}^{e},b_{t}^{e} \right) =\log \frac{\sum_{z^t}{q_t\left( \left. s_t \right|z^t \right) \int{b_{t}^{e}\left( x_t,y^{t-1},z^t \right) dx_t}}}{\left( \sum_{y^{t-1},z^t}{q_t\left( \left. s_t \right|z^t \right) \int{b_{t}^{e}\left( x_t,y^{t-1},z^t \right) dx_t}} \right) \left( \sum_{z^t}{\int{b_{t}^{e}\left( x_t,y^{t-1},z^t \right) dx_t}} \right)},
	\end{align}
	and
	\begin{align}\label{Eq.condCtrCostProof}
		\mathsf{E}\left[ \left. c\left( X_t,U_t \right) \right|S^{t-1},U^{t-1} \right] =\sum_{s_t,y^{t-1},z^t}{\int{q_t\left( s_t|z^t \right) b_{t}^{e}\left( x_t,y^{t-1},z^t \right) c\left( x_t,\mu _t\left( \gamma _{t}^{e},s_t \right) \right) dx_t}}.
	\end{align}
	with $\mathsf{E}\left[ \left. c\left( X_T \right) \right|S^{T-1},U^{T-1} \right] =\sum_{s_T,y^{T-1},z^T}{\int{q_T\left( s_T|z^T \right) b_{T}^{e}\left( x_T,y^{T-1},z^T \right) c\left( x_T \right) dx_T}}$.
	Besides, $b_{t}^e$ can be updated via
	\begin{align}
		b_{t+1}^{e}\left( x_{t+1},y^t,z^{t+1} \right) =\frac{P\left( z_{t+1}|x_{t+1} \right) P\left( y_t|y_{t-1} \right) q _{t}\left( \left. S_t \right|z^t \right) \int{p\left( x_{t+1}|x_t,y_t,\gamma _{t}^{c}\left( \gamma _{t}^{e},S_t \right) \right) b_{t}^{e}\left( x_t,y^{t-1},z^t \right) dx_t}}{\int{\sum_{y^{t-1},z^t}{q _{t}\left( \left. S_t \right|z^t \right) b_{t}^{e}\left( x_t,y^{t-1},z^t \right) dx_t}}}.
	\end{align}
	which is denoted as $b_{t+1}^e=\Phi\left( b_{t}^e,\gamma_t,S_t \right)$.
\end{lemma}
\begin{proof}	 
	See Appendix \ref{App:Lm.BeliefStateLoss}.
\end{proof}
At time $T$, given $(S^{T-1},U^{T-1})$, we define the following optimal cost-to-go function,
\begin{align}
	J_{T}^{\star}\left( S^{T-1},U^{T-1} \right) =\min_{\gamma^e _T} \mathsf{E}\left[ \left. c\left( X_T \right) +\lambda c_i\left( S_T,Y^{T-1},S^{T-1},U^{T-1} \right) \right|S^{T-1},U^{T-1} \right] ,
\end{align}
where the conditional stage costs are functions of $b^e_T$ according to Lemma \ref{Lm.BeliefStateLoss}. Therefore, we can denote $J_{T}^{\star}\left( S^{T-1},U^{T-1} \right)$ with $J_{T}^{\star}\left( b^e_T \right)$. 

We assume that the optimal cost-to-go function at $t+1$ is a function of $b_{t+1}^e$, i.e., $J_{t+1}^{\star}\left( S^{t},U^{t} \right)=J_{t+1}^{\star}\left( b^e_{t+1} \right)$. Then, at time step $t$, the optimal cost-to-go function is given as
\begin{align}
	J_{t}^{\star}\left( S^{t-1},U^{t-1} \right) =\min_{\gamma _t} \mathsf{E}\left[ \left. c\left( X_t,U_t \right) +\lambda c_i\left( S_t,Y^{t-1},S^{t-1},U^{t-1} \right) \right|S^{t-1},U^{t-1} \right] +\mathsf{E}\left[ \left. J_{t+1}^{\star}\left( S^t,U^t \right) \right|S^{t-1},U^{t-1} \right] , \nonumber
\end{align}
where the first term is a function of $b_t^e$ and $\gamma_t$, and the second term can be written as $J_{t+1}^{\star}\left( S^t,U^t \right) =J_{t+1}^{\star}\left( b_{t+1}^{e} \right) =J_{t+1}^{\star}\left( \Phi \left( b_{t}^{e},\gamma _t,S_t \right) \right)$, according to Lemma \ref{Lm.BeliefStateLoss}. The optimality equation \eqref{Eq.VFEst2} then follows by induction.
\section{Proof of Theorem \ref{Th.OPTEQU2}}\label{App:Th.OPTEQU2}
To prove Theorem \ref{Th.OPTEQU2}, we show that the stage cost admits a representation in terms of belief states and derive the corresponding update equations of beliefs.
\begin{lemma}\label{Lm.BeliefStateLoss2}
	Let $b_{t}^c$ denotes the controller's belief state, \emph{i.e.,} $b_{t}^c\left( x_t,y^{t-1},z^t \right) =p\left(\left. x_t\right|y^{t-1},z^t,U^{t-1} \right) P\left(\left. y^{t-1},z^t\right|S^t,U^{t-1} \right)$.
	Given the controller policy $\pi_{t}^c\left( S^{t}, U^{t-1}\right)$ and information $\paren{S^{t}, U^{t-1}}$, we have 
	\begin{align}
		\mathsf{E}\left[\left. c\paren{X_t,U_t}\right|S^{t}, U^{t-1} \right] = \sum_{y^{t-1},z^t}{\int{b_{t}^c\left( x_t,y^{t-1},z^t \right)  c\left( x_t,\pi_{t}^c\left( S^{t}, U^{t-1}\right) \right)dx_t}}, \nonumber
	\end{align}
	Moreover, $b_{t}^c$ can be computed via 
	\begin{equation}
		b_{t}^c\left( x_t,y^{t-1},z^t \right) =\frac{q_t\left(\left. S_t\right|z^t\right) b_{t}^e\left( x_t,y^{t-1},z^t \right)}{\sum_{y^{t-1},z^t}{q_t\left(\left. S_t\right|z^t\right) \int{b_{t}^e\left( x_t,y^{t-1},z^t \right) dx_t}}}, \nonumber
	\end{equation}
	which is denoted by $b_{t}^c=\Phi ^c\left( b_{t}^e,\gamma^e _t,S_t \right)$. $b_{t+1}^e$ can be updated via
	\begin{equation}
		b_{t+1}^e\left( x_{t+1},y^t,z^{t+1} \right)
		=\int{P\left(\left. z_{t+1}\right|x_{t+1}\right) p\left(\left. x_{t+1}\right|x_t,y_t,U_t \right) P\left(\left. y_t\right|y_{t-1} \right) b_{t}^c\left( x_t,y^{t-1},z^t \right) dx_t}, \nonumber
	\end{equation}  
	which is denoted as $b_{t+1}^e=\Phi ^e\left( b_{t}^c,U_t \right)$.
\end{lemma}
\begin{proof}	 
	See Appendix \ref{App:Lm.BeliefStateLoss2}.
\end{proof}
At time $T$, given $\left(S^{T},U^{T-1}\right)$, we define the controller’s optimal cost-to-go function as
\begin{align}
	\tilde{J}_{T}^{c,\star}\left(S^{T},U^{T-1}\right)
	=
	\mathsf{E}\left[ c\left(X_T\right)\middle|S^{T},U^{T-1} \right]. \nonumber
\end{align}
By Lemma \ref{Lm.BeliefStateLoss2}, $\tilde{J}_{T}^{c,\star}\left(S^{T},U^{T-1}\right)$ is completely determined by the belief state $b_{T}^c$, and can therefore be written as $\tilde{J}_{T}^{c,\star}\left(b_{T}^c\right)$. Moreover, using the belief update rule $
b_{T}^{c}=\Phi ^c\left( b_{T}^{e},\gamma _{T}^{e},S_T \right)$, we obtain $
\tilde{J}_{T}^{c,\star}\left(b_{T}^c\right)
=
\tilde{J}_{T}^{c,\star}\left( \Phi ^c\left( b_{T}^{e},\gamma _{T}^{e},S_T \right) \right)$.

Given $\left(S^{T-1},U^{T-1}\right)$, we define the quantizer’s optimal cost-to-go function based on $\tilde{J}_{T}^{c,\star}\left(b_{T}^c\right)$ as
\begin{align}
	\tilde{J}_{T}^{e,\star}\left( S^{T-1},U^{T-1} \right) 
	=
	\min_{\gamma _{T}^{e}}
	\mathsf{E}\left[ \lambda c_i\left( S_T,Y^{T-1},S^{T-1},U^{T-1} \right) \middle| S^{T-1},U^{T-1} \right]
	+\mathsf{E}\left[ c\left(X_T\right) \middle| S^{T-1},U^{T-1} \right]. \nonumber
\end{align}
According to \eqref{Eq.infLossProof}, the first term is determined by $b_{T}^e$. For the second term, we have
\begin{align}
	\mathsf{E}\left[ c\left( X_T \right) \middle| S^{T-1},U^{T-1} \right]
	&=
	\sum_{s_T}
	P\left( s_T\middle| S^{T-1},U^{T-1} \right)
	\mathsf{E}\left[ c\left( X_T \right) \middle| s_T,S^{T-1},U^{T-1} \right] \nonumber
	\\
	&=
	\sum_{y^{T-1},z^T,s_T}
	q_T\left( \left. s_T \right| z^T \right)
	\int b_{T}^{e}\left( x_T,y^{T-1},z^T \right) dx_T
	\,\tilde{J}_{T}^{c,\star}\left( \Phi ^c\left( b_{T}^{e},\gamma _{T}^{e},s_T \right) \right), \nonumber
\end{align}
which is also determined by $b_{T}^e$. Hence, $\tilde{J}_{T}^{e,\star}\left( S^{T-1},U^{T-1} \right)$ can be equivalently written as $\tilde{J}_{T}^{e,\star}\left( b_{T}^e \right)$.

At time $T-1$, given $\left(S^{T-1},U^{T-2}\right)$, we similarly obtain
\begin{align}
	\tilde{J}_{T-1}^{c,\star}\left( S^{T-1},U^{T-2} \right)
	&=
	\min_{U_{T-1}}
	\mathsf{E}\left[ c\left( X_{T-1},U_{T-1} \right) \middle| S^{T-1},U^{T-2} \right]
	+\tilde{J}_{T}^{e,\star}\left( S^{T-1},U^{T-1} \right) \nonumber
	\\
	&=
	\min_{U_{T-1}}
	\mathsf{E}\left[ c\left( X_{T-1},U_{T-1} \right) \middle| S^{T-1},U^{T-2} \right]
	+\tilde{J}_{T}^{e,\star}\left( b_{T}^{e} \right). \nonumber
\end{align}
By Lemma \ref{Lm.BeliefStateLoss2}, the first term is determined by $b_{T-1}^c$, while the second term satisfies $
\tilde{J}_{T}^{e,\star}\left( b_{T}^{e} \right)
=
\tilde{J}_{T}^{e,\star}\left( \Phi ^e\left( b_{T-1}^{c},U_{T-1} \right) \right)$. Therefore, $\tilde{J}_{T-1}^{c,\star}\left( S^{T-1},U^{T-2} \right)$ can be equivalently represented as $\tilde{J}_{T-1}^{c,\star}\left( b_{T-1}^c \right)$.

Proceeding in the same manner, we recursively obtain the coupled Bellman equations \eqref{Eq.VFEst} and \eqref{Eq.VFCtr} for all $t=T-1,T-2,\ldots,1$.
%After substituting \eqref{Eq.BSUPEst} into \eqref{Eq.VFEst}, the Bellman optimality equation \eqref{Eq.VFEst} depends on the belief state $b_{t}^e$. Thus, the optimal solution of \eqref{Eq.VFEst} is a function of $b_{t}^e$, \emph{i.e.,} $\gamma^{e,\star}_t\left(b_{t}^e\right)$. Similar, the optimal solution of \eqref{Eq.VFCtr} is $\pi_{t}^c\left(b_{t}^c\right)$. 
\section{Proof of Lemma \ref{Lm.InfCmp}}\label{App:Th.InfCmp}
We separate the information loss into two terms by adding a uniform random variable $\tilde{S}_t$,
\begin{equation}
	\log \frac{P\left( S_t,Y^{t-1},S^{t-1},U^{t-1} \right)}{P\left( S_t|S^{t-1},U^{t-1} \right) P\left( Y^{t-1},S^{t-1},U^{t-1} \right)}=\log \frac{P\left( S_t,Y^{t-1},S^{t-1},U^{t-1} \right)}{P\left( \tilde{S}_t \right) P\left( Y^{t-1},S^{t-1},U^{t-1} \right)}-\log \frac{P\left( S_t,S^{t-1},U^{t-1} \right)}{P\left( \tilde{S}_t \right) P\left( S^{t-1},U^{t-1} \right)}.
\end{equation}
Next, we show each logarithmic term could be computed using a binary classifier.

Let $M_t=1$ denote positive samples from the space $\left\{S_t,Y^{t-1},S^{t-1},U^{t-1}\right\}$, and $M_t=0$ denote negative samples from $\left\{\tilde{S}_t,Y^{t-1},S^{t-1},U^{t-1}\right\}$. We use $\left\{\bar{S}_t,Y^{t-1},S^{t-1},U^{t-1}\right\}$ to denote combined sample space, where $\bar{S}_t=M_tS_t+\left( 1-M_t \right) \tilde{S}_t$. Since $\tilde{S}_t$ is random variable which is independent from $S_t,Y^{t-1},U^{t-1},S^{t-1}$, we could assume $P\left( M_t=0 \right) =P\left( M_t=1 \right) =0.5$.

Then we could rewrite the likelihood ratio as
\begin{equation}
	\log \frac{P\left( S_t,Y^{t-1},S^{t-1},U^{t-1} \right)}{P\left( \tilde{S}_t \right) P\left( Y^{t-1},S^{t-1},U^{t-1} \right)}=\log \frac{P\left( \bar{S}_t,Y^{t-1},S^{t-1},U^{t-1}|M_t=1 \right)}{P\left( \bar{S}_t,Y^{t-1},S^{t-1},U^{t-1}|M_t=0 \right)}.
\end{equation}
With, 
\begin{equation}
	P\left( \bar{S}_t,Y^{t-1},S^{t-1},U^{t-1}|M_t=1 \right) = \frac{P\left( M_t=1|\bar{S}_t,Y^{t-1},S^{t-1},U^{t-1} \right) P\left( \bar{S}_t,Y^{t-1},S^{t-1},U^{t-1} \right)}{P\left( M_t=1 \right)},
\end{equation}
\begin{equation}
	P\left( \bar{S}_t,Y^{t-1},S^{t-1},U^{t-1}|M_t=0 \right) = \frac{P\left( M_t=0|\bar{S}_t,Y^{t-1},S^{t-1},U^{t-1} \right) P\left( \bar{S}_t,Y^{t-1},S^{t-1},U^{t-1} \right)}{P\left( M_t=0 \right)}.
\end{equation}
we have
\begin{equation}
	\log \frac{P\left( S_t,Y^{t-1},S^{t-1},U^{t-1} \right)}{P\left( \tilde{S}_t \right) P\left( Y^{t-1},S^{t-1},U^{t-1} \right)}=\log \frac{P\left( M_t=1|\bar{S}_t,Y^{t-1},S^{t-1},U^{t-1} \right)}{P\left( M_t=0|\bar{S}_t,Y^{t-1},S^{t-1},U^{t-1} \right)}.
\end{equation}
Let $w_{t}^{\star}\left( \bar{S}_t,Y^{t-1},S^{t-1},U^{t-1} \right) =P\left( M_t=1|\bar{S}_t,Y^{t-1},S^{t-1},U^{t-1} \right) $, we have
\begin{equation}
	\log \frac{P\left( S_t,Y^{t-1},S^{t-1},U^{t-1} \right)}{P\left( \tilde{S}_t \right) P\left( Y^{t-1},S^{t-1},U^{t-1} \right)}=\log \frac{w_{t}^{\star}\left( \bar{S}_t,Y^{t-1},S^{t-1},U^{t-1} \right)}{1-w_{t}^{\star}\left( \bar{S}_t,Y^{t-1},S^{t-1},U^{t-1} \right)}.
\end{equation}
According to \cite{molavipour2021neural}, the optimal classifier $P\left( M_t|\bar{X}_t,Y^{t-1},U^{t-1},S^{t-1} \right) $ could be obtained via the cross entropy loss optimization \eqref{Eq.likeRatioOP1}.
Similarly, we could prove that
\begin{equation}
	\log \frac{P\left( S_t,S^{t-1},U^{t-1} \right)}{P\left( \tilde{S}_t \right) P\left( S^{t-1},U^{t-1} \right)}=\log \frac{\xi _{t}^{\star}\left( \bar{S}_t,S^{t-1},U^{t-1} \right)}{1-\xi _{t}^{\star}\left( \bar{S}_t,S^{t-1},U^{t-1} \right)},
\end{equation}
where $\xi_{t}^{\star}$ is the optimal solution of  \eqref{Eq.likeRatioOP2}.

\section{Proof of Theorem \ref{Th.gradientObj}}\label{App:Th.gradientObj}
The gradient of $\mathsf{E}_{\theta}\left[ \sum_{t=1}^T{c\left( X_t,U_t \right)} \right] $ with respect to $\theta$ is
\begin{align}\label{Eq:cost-Grad}
	\nabla _{\theta}\mathsf{E}_{\theta}\left[ \sum_{t=1}^T{c\left( X_t,U_t \right)} \right] =\mathsf{E}_{\theta}\left[ \left( \sum_{t=1}^T{c\left( X_t,U_t \right)} \right) \nabla _{\theta}\log P_{\theta}\left( \tau \right) \right] , 
\end{align}
where $\tau =\left( Y^T,X^T,Z^T,S^T,U^{T} \right)$ and $\nabla _{\theta}\log P_{\theta}\left( \tau \right)$ can be simplified via
\begin{align}
	\nabla _{\theta}\log P_{\theta}\left( \tau \right) &=\nabla _{\theta}\log \prod_{t=1}^T{P\left( Y_t|Y_{t-1} \right) p\left( X_t|X_{t-1},Y_{t-1},U_{t-1} \right) P\left( Z_t|X_t \right) \pi _{\theta}\left( \left. U_t,S_t \right|Z^t, S^{t-1},U^{t-1} \right)} \nonumber \\ &=\sum_{t=1}^T{\nabla _{\theta}\log \pi _{\theta}\left( \left. U_t,S_t \right|Z^t, S^{t-1},U^{t-1} \right)}. \nonumber
\end{align}

The gradient of $\mathsf{E}_{\theta}\left[ \sum_{t=1}^T{c_{i,\theta}\left( S_t,Y^{t-1},S^{t-1},U^{t-1} \right)} \right]$ with respect to $\theta$ is
\begin{align}\label{Eq:MI-Grad}
	&\nabla _{\theta}\mathsf{E}_{\theta}\left[ \sum_{t=1}^T{c_{i,\theta}\left( S_t,Y^{t-1},S^{t-1},U^{t-1} \right)} \right] \nonumber \\
	=&\mathsf{E}_{\theta}\left[ \left( \sum_{t=1}^T{c_{i,\theta}\left( S_t,Y^{t-1},S^{t-1},U^{t-1} \right)} \right) \nabla _{\theta}\log P_{\theta}\left( \tau \right) \right] +\mathsf{E}_{\theta}\left[ \sum_{t=1}^T{\nabla _{\theta}c_{i,\theta}\left( S_t,Y^{t-1},S^{t-1},U^{t-1} \right)} \right] .  
\end{align}
The second term in the right-hand side of \eqref{Eq:MI-Grad} is 0, since it can be computed via
\begin{align}
	\mathsf{E}_{\theta}\left[ \nabla _{\theta}c_{i,\theta}\left( S_t,Y^{t-1},S^{t-1},U^{t-1} \right) \right] =\mathsf{E}_{\theta}\left[ \nabla _{\theta}\log \frac{P_{\theta}\left( S_t \middle| Y^{t-1},S^{t-1},U^{t-1} \right)}{P_{\theta}\left( S_t \middle| S^{t-1},U^{t-1} \right)} \right] ,
\end{align}
where
\begin{align}
	\mathsf{E}_{\theta}\left[ \nabla _{\theta}\log P_{\theta}\left( S_t \middle| Y^{t-1},S^{t-1},U^{t-1} \right) \right]& =\mathsf{E}_{\theta}\left[ \sum_{s_t}{P_{\theta}\left( s_t \middle| Y^{t-1},S^{t-1},U^{t-1} \right)}\frac{\nabla _{\theta}P_{\theta}\left( s_t \middle| Y^{t-1},S^{t-1},U^{t-1} \right)}{P_{\theta}\left( s_t \middle| Y^{t-1},S^{t-1},U^{t-1} \right)} \right]  \nonumber
	\\
	&=\mathsf{E}_{\theta}\left[ \nabla _{\theta}\sum_{s_t}{P_{\theta}\left( s_t \middle| Y^{t-1},S^{t-1},U^{t-1} \right)} \right] \nonumber
	\\
	&=\mathsf{E}_{\theta}\left[ \nabla _{\theta}1 \right] \nonumber
	\\
	&=0, \nonumber
\end{align}
and $\mathsf{E}_{\theta}\left[ \nabla _{\theta}\log P_{\theta}\left( S_t \middle| S^{t-1},U^{t-1} \right) \right] =0$ similarly.

Therefore, the gradient of $L_\theta$ with respect to $\theta$ is \eqref{Eq.gradObj}.
\section{Proof of Lemma \ref{Lm.MI}}\label{App:Lm.MI}
Let $S_{0}=U_{0}=\emptyset$. We next show the mutual information could be decomposed into a conditional mutual information series.
\begin{align}
	I\paren{S^T, U^{T}; Y^T} & \overset{(a)}{=} \sum_{t=1}^{T} I\paren{\left. S_t, U_t; Y^T \right| S^{t-1}, U^{t-1}} \nonumber\\
	& \overset{(b)}{=} \sum_{t=1}^{T} I\paren{\left. S_t, U_t; Y^{t-1} \right| S^{t-1}, U^{t-1}} + I\paren{\left. S_t, U_t; Y^T_t \right| S^{t-1}, U^{t-1}, Y^{t-1}} \nonumber
\end{align}
where $(a)$ and $(b)$ follow from the chain rule. We next show the second term is 0. Under the causal information structure, $S_t$ is generated from $(Z^t,S^{t-1},U^{t-1})$, and $U_t$ is generated from $(S^t,U^{t-1})$. Moreover, $X_t$ and $Z^t$ depend only on $Y^{t-1}$, while the future private process $Y_t^T$ is exogenous and conditionally independent of $(X_t,Z^t,S_t,U_t)$ given $Y^{t-1}$. Hence,
\[
P(y_t^T,s_t,u_t \mid s^{t-1},u^{t-1},y^{t-1})
=
P(y_t^T \mid s^{t-1},u^{t-1},y^{t-1})
P(s_t,u_t \mid s^{t-1},u^{t-1},y^{t-1}),
\]
which implies
\[
I(S_t,U_t;Y_t^T \mid S^{t-1},U^{t-1},Y^{t-1})=0.
\]

Since $U_t$ is independent of $Y^{t-1}$ given $\left(S^t,U^{t-1}\right)$, we have $I\paren{\left. S_t, U_t; Y^{t-1} \right| S^{t-1}, U^{t-1}} = I\paren{\left. S_t; Y^{t-1} \right| S^{t-1}, U^{t-1}}$, therefore,
$$\min_{\pi^e, \pi^c} L\left(\pi^e, \pi^c\right)= \min_{\left\{ \pi _{t}^e,\pi _{t}^c \right\} _{t=1}^{T}} \sum_{t=1}^{T}{\mathsf{E}\left[ c\left( X_t,U_t \right) \right]} +\lambda I\left( \left. S_t;Y^{t-1} \right| S^{t-1}, U^{t-1}\right).$$
\section{Proof of Lemma \ref{Lm.BeliefStateLoss}}\label{App:Lm.BeliefStateLoss}
Given $\paren{S^{t-1},U^{t-1}}$ and the collection $\gamma^e_t$, the conditional mutual information is given by
\begin{equation}
	\begin{aligned}
		\mathsf{E}\left[ \left. c_i\left( S_t,Y^{t-1},S^{t-1},U^{t-1} \right) \right|S^{t-1},U^{t-1} \right] =\sum_{s_t,y^{t-1}}{P\left( \left. s_t,y^{t-1} \right|S^{t-1} \right) \log \frac{P\left( \left. s_t,y^{t-1} \right|S^{t-1},U^{t-1} \right)}{P\left( \left. s_t \right|S^{t-1},U^{t-1} \right) P\left( \left. y^{t-1} \right|S^{t-1},U^{t-1} \right)}}, \nonumber
	\end{aligned}
\end{equation}
where each probability satisfies
$$P\left(\left. s_t,y^{t-1}\right|S^{t-1}, U^{t-1} \right) =\sum_{z^t}{q_t\left(\left. s_t\right|z^t\right) P\left(\left. y^{t-1},z^t\right|S^{t-1}, U^{t-1} \right)}=\sum_{z^t}{q_t\left( \left. s_t\right|z^t \right) \int{b_{t}^e\left(x_t, y^{t-1},z^t \right)dx_t}},$$

$$P\left(\left. s_t\right|S^{t-1}, U^{t-1} \right) =\sum_{y^{t-1},z^t}{q_t\left(\left. s_t\right|z^t\right) P\left(\left. y^{t-1},z^t\right|S^{t-1}, U^{t-1} \right)}=\sum_{y^{t-1},z^t}{q_t\left( \left. s_t\right|z^t \right) \int{b_{t}^e\left(x_t, y^{t-1},z^t \right)dx_t}},$$

$$P\left(\left. y^{t-1}\right|S^{t-1}, U^{t-1} \right) =\sum_{z^t}{P\left(\left. y^{t-1},z^t\right|S^{t-1}, U^{t-1} \right)}=\sum_{z^t}{\int{b_{t}^e\left(x_t, y^{t-1},z^t \right)dx_t}}.$$
Therefore, we have \eqref{Eq.condMIProof}. Similarly, the stage cost $\mathsf{E}\left[ \left. c\left( X_t,U_t \right) \right|S^{t-1},U^{t-1} \right]$ can be computed via \eqref{Eq.condCtrCostProof}.

Besides, the belief state $b_{t+1}^e$ can be computed via
\begin{align}
	&b_{t+1}^{e}\left( x_{t+1},y^t,z^{t+1} \right) \nonumber \\
	&=p\left( \left. x_{t+1} \right|y^t,z^{t+1},U^t \right) P\left( \left. y^t,z^{t+1} \right|S^t,U^t \right) \nonumber
	\\
	&=\frac{\int{p\left( \left. x_{t+1},x_t \right|y^t,z^{t+1},U^t \right) P\left( \left. S_t, U_t,y^t,z^{t+1} \right|S^{t-1},U^{t-1} \right) dx_t}}{P\left( \left. S_t, U_t \right|S^{t-1},U^{t-1} \right)} \nonumber
	\\
	&\overset{\left( a \right)}{=}\frac{\int{P\left( z_{t+1}|x_{t+1} \right) P\left( y_t|y_{t-1} \right) p\left( \left. x_{t+1} \right|x_t,y_t,U_t \right) P\left( \left. S_t \right|z^t,S^{t-1},U^{t-1} \right) p\left( \left. x_t \right|y^{t-1},z^t,U^{t-1} \right) P\left( \left. y^{t-1},z^t \right|S^{t-1},U^{t-1} \right) dx_t}}{P\left( \left. S_t, U_t \right|S^{t-1},U^{t-1} \right)} \nonumber
	\\
	&\overset{\left( b \right)}{=}\frac{P\left( z_{t+1}|x_{t+1} \right) P\left( y_t|y_{t-1} \right) q_t\left( \left. S_t \right|z^t \right) \int{p\left( x_{t+1}|x_t,y_t,\gamma _{t}^{c}\left( \gamma _{t}^{e},S_t \right) \right) b_{t}^{e}\left( x_t,y^{t-1},z^t \right) dx_t}}{\int{\sum_{y^{t-1},z^t}{q_t\left( \left. S_t \right|z^t \right) b_{t}^{e}\left( x_t,y^{t-1},z^t \right) dx_t}}}, \nonumber
\end{align}
where $(a)$ follows the system dynamics, $(b)$ follows from the fact that the denominator in $(a)$ is a marginal distribution of the numerator. Notably, $b_{t+1}^e$ depends on the policy collections $\gamma^e_t$, $\gamma^c_t$ and the quantization index $S_t$, therefore, we denote this update by $b_{t+1}^e=\Phi\left( b_{t}^e,\gamma_t,S_t \right)$.

\section{Proof of Lemma \ref{Lm.BeliefStateLoss2}}\label{App:Lm.BeliefStateLoss2}
Using the definitions of $b_{t}^e$ and $b_{t}^c$, we compute $b_{t}^c$ based on $b_{t}^e$ as follows, 
\begin{align}
	b_{t}^c\left( x_t,y^{t-1},z^t \right) &= p\left(\left. x_t\right|y^{t-1},z^t,U^{t-1} \right) P\left(\left. y^{t-1},z^t\right|S^t,U^{t-1} \right) \nonumber\\
  &= \frac{p\left(\left. x_t\right|y^{t-1},z^t,U^{t-1} \right) P\left(\left. y^{t-1},z^t, S_t \right|S^{t-1},U^{t-1} \right)}{P\left(\left. S_t \right|S^{t-1},U^{t-1} \right)} \nonumber\\
  &\overset{(a)}{=} \frac{q_t\paren{\left. S_t \right| z^t} p\left(\left. x_t\right|y^{t-1},z^t,U^{t-1} \right) P\left(\left. y^{t-1},z^t\right|S^{t-1},U^{t-1} \right)}{\sum_{y^{t-1},z^t}\int{q_t\paren{\left. S_t \right| z^t} p\left(\left. x_t\right|y^{t-1},z^t,U^{t-1} \right) P\left(\left. y^{t-1},z^t\right|S^{t-1},U^{t-1} \right)dx_t}} \nonumber\\
  &\overset{(b)}{=} \frac{q_t\paren{\left. S_t \right| z^t} b_{t}^e\paren{x_t,y^{t-1},z^t}}{\sum_{y^{t-1},z^t}\int{q_t\paren{\left. S_t \right| z^t} b_{t}^e\paren{x_t,y^{t-1},z^t}dx_t}}. \nonumber
\end{align}
where the numerator in $\paren{a}$ is derived by extracting $q_t\paren{\left. S_t \right| z^t}$ from $P\left(\left. y^{t-1},z^{t},S_t\right|S^{t-1}, U^{t-1} \right)$, and $(b)$ follows from the definition of $b_{t}^e$.
Also, we compute $b_{t+1}^e$ with $b_{t}^c$ as follows,
\begin{align} 
	&b_{t+1}^e\left( x_{t+1},y^{t},z^{t+1} \right) \nonumber\\
	=& p\left( \left. x_{t+1}\right|y^{t},z^{t+1},U^{t} \right) P\left(\left. y^{t},z^{t+1}\right|S^{t},U^{t} \right) \nonumber \\
	\overset{(a)}{=}& P\paren{\left. z_{t+1} \right| x_{t+1}} \int{p\paren{\left. x_{t+1}, x_t\right|y^{t},z^{t},U^{t}}dx_t} P\left(\left. y^{t},z^{t}\right|S^{t},U^{t} \right) \nonumber\\
	\overset{(b)}{=}& P\paren{\left. z_{t+1} \right| x_{t+1}} \int{p\left(\left. x_{t+1}\right|x_t,y_t,U_t \right) p\paren{\left. x_t\right|y^{t-1},z^{t},U^{t}}dx_t} P\left(\left. y^{t-1},z^{t}\right|S^{t},U^{t} \right)P\paren{\left. y_t\right| y^{t-1}, z^t, S^t, U^t} \nonumber\\
	\overset{(c)}{=}& P\paren{\left. z_{t+1} \right| x_{t+1}} \int{p\left(\left. x_{t+1}\right|x_t,y_t,U_t \right) p\paren{\left. x_t\right|y^{t-1},z^{t},U^{t-1}}dx_t} P\left(\left. y^{t-1},z^{t}\right|S^{t},U^{t-1} \right) P\paren{\left. y_t\right| y_{t-1}} \nonumber\\
	\overset{(d)}{=}&\int{P\left(\left. z_{t+1}\right|x_{t+1} \right) p\left(\left. x_{t+1}\right|x_t,y_t,U_t \right) P\left(\left. y_t\right|y_{t-1} \right) b_{t}^c\left( x_t,y^{t-1},z^t \right) dx_t}, \nonumber
\end{align}
where $(a)$ and $(b)$ follow the system dynamics, $(c)$ is derived by removing irrelevant conditional information, $(d)$ follows from the definition of $b_{t}^c$.

Given $\paren{S^{t},U^{t-1}}$ and the controller $U_t=\pi_t^c\paren{S^t,U^{t-1}}$, the conditional stage cost is given by
\begin{align}
	\mathsf{E}\left[\left. c\paren{X_t,U_t}\right|S^{t}, U^{t-1} \right] &= \sum_{y^{t-1},z^t}{\int{p\paren{\left. x_t \right| S^t,U^{t-1}} c\left( x_t,U_t \right)dx_t}} \nonumber \\ 
																			   &\overset{(a)}{=} \sum_{y^{t-1},z^t}{\int{p\left(\left. x_t\right|y^{t-1},z^t,S^t,U^{t-1} \right) P\left(\left. y^{t-1},z^t\right|S^t,U^{t-1} \right) c\left( x_t,U_t \right)dx_t}} \nonumber \\
																			   &\overset{(b)}{=} \sum_{y^{t-1},z^t}{\int{b_{t}^c\left( x_t,y^{t-1},z^t \right) c\left( x_t,U_t \right)dx_t}}, \nonumber
\end{align}
where $(a)$ follows from the chain rule, $(b)$ follows from the Markov chain $X_t \rightarrow \paren{Z^t, U^{t-1}} \rightarrow S^t$ and the definition of $b_{t}^c$.
\end{document}